\documentclass[a4paper, 12pt]{article}

\usepackage[a4paper,top=2cm,bottom=4cm,left=2.5cm,right=2.5cm,marginparwidth=1.75cm]{geometry}

\usepackage{mathtools}
\usepackage{amsfonts} 
\usepackage{amsmath} 
\usepackage{amssymb} 
\usepackage{amsthm} 
\usepackage[english]{babel}
\usepackage{bbm}
\usepackage{bm}  
\usepackage{anyfontsize}
\usepackage{booktabs} 

\usepackage[font={small,it}]{caption} 
\usepackage{comment} 
\usepackage{enumitem} 
\usepackage{esvect}   
\usepackage{etoolbox} 
\usepackage[T1]{fontenc}
\usepackage{graphics} 

\usepackage[utf8]{inputenc}
\usepackage{latexsym} 

\usepackage{listofitems} 
\usepackage{mathrsfs} 
\usepackage{mathtools} 
\usepackage{multicol} 
\usepackage{multirow} 

\usepackage[authoryear]{natbib}
\usepackage{tikz}
\usepackage{url} 
\usepackage{xcolor}
\usepackage{xr}
\usepackage{ifthen}

\usepackage[authoryear]{natbib}
\usepackage{makeidx}         
\usepackage{graphicx}        
\usepackage[colorlinks=true, allcolors=blue]{hyperref}
\usepackage{subcaption}
\usepackage{booktabs}

\newcommand{\de}{\mathrm{d}}
\newcommand{\R}{\mathbb{R}}
\newcommand{\N}{\mathbb{N}}
\newcommand{\p}{\mathbb{P}_0^p}
\newcommand{\ptheta}{\mathbb{P}_\vartheta^p}
\newcommand{\Psib}{\bm{\Psi}}
\renewcommand{\S}{{\cal S}}
\newcommand{\F}{{\cal D}} 
\newcommand{\HS}{{Hilbert-Schmidt }}
\newcommand{\kronecker}{\otimes}
\newcommand{\tX}{{\tilde X}}
\newcommand{\tony}[1]{{\color{teal}#1}}
\DeclarePairedDelimiter\ceil{\lceil}{\rceil}
\DeclarePairedDelimiter\floor{\lfloor}{\rfloor}
\newcommand{\cng}[1]{{\color{orange}#1}}
\newcommand{\rd}[1]{{\color{red}#1}}
\newcommand{\Aset}{\mathcal{A}}
\newcommand{\sgn}{\mathrm{s}}
\newcommand{\eexp}{\mathrm{exp}}
\newcommand{\betahat}{\widehat{\boldsymbol\beta}}
\newcommand{\psihat}{\widehat{\boldsymbol\psi}}
\newcommand{\thetahat}{\widehat{\boldsymbol\theta}}
\newcommand{\trace}{\mathrm{trace}}
\newcommand{\vecof}{\mathrm{vec}}
\DeclareMathOperator{\diag}{diag}
\newcommand{\T}{\top}

\def\bbeta{\boldsymbol{\beta}}

\def\bepsilon{\boldsymbol{\epsilon}}
\def\bvarepsilon{\boldsymbol{\varepsilon}}

\def\bPi{\boldsymbol{\Pi}}

\def\bPhi{\boldsymbol{\Phi}}

\def\b0{\boldsymbol{0}}

\def\ba{\boldsymbol{a}}
\def\bb{\boldsymbol{b}}

\def\be{\boldsymbol{e}}

\def\bm{\boldsymbol{m}}

\def\bp{\boldsymbol{p}}

\def\bu{\boldsymbol{u}}
\def\bv{\boldsymbol{v}}
\def\bw{\boldsymbol{w}}
\def\bx{\boldsymbol{x}}
\def\by{\boldsymbol{y}}

\def\bA{\mathbf{A}}
\def\bB{\mathbf{B}}
\def\bC{\mathbf{C}}

\def\bH{\mathbf{H}}
\def\bI{\mathbf{I}}

\def\bM{\mathbf{M}}

\def\bP{\mathbf{P}}

\def\bR{\mathbf{R}}
\def\bS{\mathbf{S}}

\def\bW{\mathbf{W}}
\def\bX{\mathbf{X}}

\def\ba{\mathbf{a}}
\def\bb{\mathbf{b}}

\def\be{\mathbf{e}}

\def\bm{\mathbf{m}}

\def\bp{\mathbf{p}}

\def\bu{\mathbf{u}}
\def\bv{\mathbf{v}}
\def\bw{\mathbf{w}}
\def\bx{\mathbf{x}}
\def\by{\mathbf{y}}

\def\tX{\textrm{\textup{X}}}

\newtheorem{theorem}{Theorem}
\newtheorem{corollary}{Corollary}
\newtheorem{lemma}{Lemma}

\hypersetup{ colorlinks=true, linkcolor=black, filecolor=black, urlcolor=black }

\usepackage{xr}
\begin{document}
\title{\LARGE \textbf{Degrees of Freedom in Penalized Regression:\\ Model Selection with Adaptive Penalties}} 

\author{
  Mauro Bernardi\\
  \small{Department of Statistical Sciences, University of Padova, Padova, Italy}\\
  \small{\texttt{mauro.bernardi@unipd.it}}
  \and
  Antonio Canale\\
  \small{Department of Statistical Sciences, University of Padova, Padova, Italy}\\
  \small{\texttt{antonio.canale@unipd.it}}
  \and
  Marco Stefanucci\\
  \small{Department of Economics and Finance, University of Rome Tor Vergata, Rome, Italy}\\
  \small{\texttt{marco.stefanucci@uniroma2.it}}
}

\date{}

\maketitle

\begin{abstract}
Model selection in penalized regression critically depends on an accurate assessment of model complexity, commonly quantified through the effective degrees of freedom. While the Lasso admits a simple and unbiased characterization, given by the size of the active set, this property does not extend to adaptive penalization methods, despite the widespread use of this approximation in practice. To solve this issue, in this paper we derive a novel unbiased estimator of the effective degrees of freedom for the Adaptive Lasso within Stein’s unbiased risk estimation framework. Our analysis reveals additional terms induced by data-dependent penalization, reflecting the role of adaptive weights and regularization in determining model complexity. We further revisit the Group Lasso, providing an alternative derivation of its degrees of freedom, and extend these results to the Adaptive Group Lasso.
Importantly, we characterize the behavior of the degrees of freedom along the regularization path beyond the orthonormal design setting commonly assumed in the literature, providing a new theoretical description of this behavior under general design matrices.
%
By correcting the common misuse of active set size as a proxy for degrees of freedom, our results enable more reliable risk estimation and inference, offering a rigorous foundation for understanding model complexity in adaptive penalized regression. 
\end{abstract} 


\section{Introduction}
\label{sec:mod}
Model selection is a central component of modern statistical modeling, particularly in regression settings where regularization is used to balance predictive accuracy and interpretability. In penalized regression, estimators are obtained by minimizing an objective function combining a loss term with a penalty on the regression coefficients, and the choice of the regularization parameter directly determines the complexity of the fitted model. A key ingredient in this process is the quantification of model complexity, typically expressed through the notion of degrees of freedom, which governs the trade-off between fit and flexibility. Among penalized estimators, the Lasso \citep{tibshirani1996} plays a prominent role due to its ability to perform variable selection via $\ell_1$ regularization. A key feature of the Lasso is that its degrees of freedom admit a simple unbiased estimate given by the size of the active set \citep{zou2007}. This result has made model selection particularly convenient in Lasso problems. However, this simplicity has also encouraged a widespread but unjustified practice: using the size of the active set as a proxy for the degrees of freedom in more general penalized regression methods. This shortcut becomes problematic in the presence of adaptive penalties. Methods such as the Adaptive Lasso \citep{zou.2006} and the Adaptive Group Lasso \citep{wangleng2008csda} introduce data-driven weights to reduce shrinkage bias and achieve improved estimation properties, including oracle behavior under suitable conditions. While these methods offer clear advantages over their non-adaptive counterparts, their data-dependent penalization fundamentally alters the relationship between model fit and complexity. As a result, the degrees of freedom can no longer be characterized by the active set alone, and naive extensions of the Lasso formula may lead to distorted model assessment and unreliable model selection in practice.\par
Formally, the concept of degrees of freedom underpins the definition of information criteria such as the Akaike Information Criterion (AIC) and Bayesian Information Criterion (BIC) combining a model fit measure with a penalty proportional to the degrees of freedom \citep{safken_etal.2025bmtrk}. 
Generally speaking, the degrees of freedom can be interpreted as the number of independent pieces of information used to estimate the unknown parameters. 
In a setting where $\bX$ is a matrix of covariates with $p$ columns, $\by$ is a response vector of size $n$, and $\bbeta$ is a $p$-dimensional vector of unknown coefficients with $\by = \bX \bbeta + \bepsilon$ and $\bepsilon \sim \text{N}(\boldsymbol{0}, \sigma^2 \bI_n)$, the degrees of freedom of the least-squares estimator equal $p$, implying that each estimated coefficient consumes one degree of freedom. A general notion of degrees of freedom is given by \cite{efron_dof}:
\begin{equation}
\label{eq:df_general_expression}
df(\widehat{\by}) = \frac{1}{\sigma^2}\sum_{i=1}^n \mathbb{C}\text{ov}(y_i, \widehat{y}_i),
\end{equation}
where $\widehat{\by} = (\widehat{y}_1, \ldots, \widehat{y}_n)$ denotes the vector of fitted values and $\sigma^2$ is the (unknown) variance of the response. Both $\by$ and $\widehat{\by}$ are random quantities, and this expression shows that stronger adaptation of the fitted values to the data corresponds to larger degrees of freedom. An equivalent formulation in terms of divergence leads to Stein’s lemma \citep{stein1981sure}:
\begin{equation}
\label{eq:df_stein}
df(\widehat{\by}) = \mathbb{E}\bigg[\text{trace} \bigg(\frac{\partial \widehat{\by}}{\partial \by}\bigg)\bigg].
\end{equation}
Under this framework, \cite{zou2007} showed that an unbiased estimate of the degrees of freedom of the Lasso is given by the number of nonzero coefficients, i.e., the size of the active set. Although the Lasso uses the response variable to perform variable selection, the resulting increase in flexibility is exactly offset by the shrinkage imposed on the coefficients. This balance does not hold for subset selection procedures, where model choice involves a discrete search and typically yields larger degrees of freedom \citep[see][]{tibshirani2015ss}. Similar results have been obtained for the generalized Lasso \citep{tibshirani_taylor_2011, Chen_etal.2020jasa} and in high-dimensional settings with $p>n$ \citep{dossal2013dof_lasso_general, tibshirani_taylor.2012}. In contrast, this exact trade-off does not extend to the Group Lasso, as shown by \cite{vaiter_etal.2017}, who derive degrees-of-freedom expressions generally smaller than the size of the active set \citep[see also][]{vaiter2012dof_grouplasso}.\par
In this paper, we address this gap by deriving the first unbiased estimator of the effective degrees of freedom for the Adaptive Lasso within Stein’s unbiased risk estimation framework \citep{stein1981sure}. Our estimator explicitly accounts for the dependence of the adaptive penalty on data-driven weights, regularization parameters, and coefficient signs. The derivation extends the analysis of \cite{zou2007} to adaptive settings, where additional terms arise due to the data-dependent penalization. These additional terms, absent in standard Lasso, are the key feature to characterize the degrees of freedom in adaptive procedures. 
As a further relevant contribution, we revisit the Group Lasso characterizing the behavior of the degrees of freedom along the regularization path beyond the orthonormal design setting commonly assumed in the literature. In particular, we show that structural properties established under orthogonality, such as monotonicity, do not generally extend to non-orthogonal designs and provide a novel characterization of this behavior under general design matrices, offering new insight that complements existing results \citep{yuan:2006:glasso, vaiter_etal.2017}.

Leveraging these results, we obtain analogous expressions for the Adaptive Group Lasso. By providing valid degrees-of-freedom estimators for adaptive penalized regression, our work enables principled information-criterion-based model selection and improves risk estimation and inference in settings where current practice often relies on unjustified simplifications.

\subsection{Notation and organization of the paper}
Throughout the paper, we assume a normal linear regression model $\by = \bX \bbeta + \bepsilon$ with $\bepsilon \sim \text{N}(\boldsymbol{0}, \sigma^2 \bI_n)$, $\bX \in \mathbb{R}^{n \times p}$, $\by \in \mathbb{R}^n$, and $\bbeta \in \mathbb{R}^p$, with $n>p$ and $\bX$ full column rank. We refer to the orthonormal design whenever $\bX^\top \bX = \bI_p$ i.e., variables are uncorrelated and standardized.
The least squares estimator of $\bbeta$ is given by $\widehat{\bbeta}^{\mathsf{LS}}=(\bX^\top\bX)^{-1}\bX^\top\by$.
Let $\gamma \in \mathbb{R}_+$ denote the regularization parameter common to the penalized methods under study. For a fixed $\gamma$,  we denote by $\widehat{\bbeta}_\gamma$, $\widehat{\by}_\gamma$, and $\widehat{df}_\gamma$ the vector of estimated coefficients, the predictions, and the estimated degrees of freedom, respectively. 
When unambiguous, we may suppress the dependence on $\gamma$ in the notation (e.g., $\widehat{\bbeta}$ for $\widehat{\bbeta}_\gamma)$ to simplify exposition. The \emph{transition points}, i.e. the values of $\gamma$ for which the active set changes are denoted by $\gamma_l$ with $l=1, \dots, L$ where $L$ is the number of such points.
Information criteria are defined as 
\begin{equation}
\label{eq:BIC}
     AIC(\widehat{\by}_\gamma) = \frac{\lVert \by - \widehat{\by}_\gamma \rVert_2^2}{n\sigma^2} + \frac{2}{n} \widehat{df}_\gamma, \quad \quad 
    BIC(\widehat{\by}_\gamma) = \frac{\lVert \by - \widehat{\by}_\gamma \rVert_2^2}{n\sigma^2} + \frac{\text{log}(n)}{n} \widehat{df}_\gamma.   
\end{equation}
The operator $|\cdot|$ denotes either the absolute value of a scalar, i.e., $|a|$ for $a \in \mathbb{R}$, or the cardinality of a set, i.e., $|\mathcal{A}|$ for a set $\mathcal{A}$.
For a square matrix $\bM$, $\diag(\bM)$ extracts the vector of its diagonal elements. For a vector $\bv = (v_1, \dots, v_p)^\top$, $\diag(\bv)$ constructs a diagonal matrix with entries $v_1, \dots, v_p$. The operator $\trace(\cdot)$ denotes the trace of a square matrix. The function $\sgn(\cdot)$ denotes the sign of a scalar or a vector.\par
The remainder of the paper is structured as follows. Section~\ref{sec:theory} presents our core theoretical contributions: closed-form expressions for the effective degrees of freedom in Adaptive Lasso, Group Lasso, and Adaptive Group Lasso, accompanied by interpretative remarks and complementary characterizations. Section~\ref{sec:exper} provides empirical validation of these results and illustration of their use for variable selection through synthetic and real-data experiments, and section~\ref{sec:conclusions} summarizes the methodological and theoretical implications of our work, along with potential extensions. The Proofs of all the results are collected in the Appendix, while proofs of  auxiliary lemmas, corollaries and additional technical results are collected in the Supplementary Materials.

\section{Main results}
\label{sec:theory}
\subsection{The degrees of freedom of Adaptive Lasso}
\label{sec:adalasso}
Let $\bX$, $\by$, $\bbeta$ and $\gamma$ as in the previous section, and $\bw \in \mathbb{R}_+^p$ a vector of non-negative weights. Adaptive Lasso \citep{zou.2006} estimator is defined as 
\begin{equation}
\label{eq:convex_regularized_problem_adalasso}
\widehat{\bbeta} = \arg\min_{\bbeta} \frac{1}{2} \lVert {\by}  - {\bX}\boldsymbol{\beta} \rVert_2^2 + \gamma \sum_{j=1}^{p} w_{j} |\beta_j|. 
\end{equation}
The weights are chosen in a data-dependent way, often being a continuous positive function of the least squares estimates, $w_j=f(\widehat{\beta}^{\mathsf{LS}}_j)$. 
Typical choices are $f(\widehat{\beta}^{\mathsf{LS}}_j) = 1/|\widehat{\beta}^{\mathsf{LS}}_j|^\alpha$ or $f(\widehat{\beta}^{\mathsf{LS}}_j) = \text{exp}(-\alpha|\widehat{\beta}^{\mathsf{LS}}_j|)$, with $\alpha>0$. Let $\mathcal{G}=\{1,\dots,p\}$. The active set, i.e. the index set of nonzero estimates, sorted in increasing order, is denoted by
\begin{equation*}
\label{eq:dof_adalasso_Aset_def}
\Aset=\left\{j\in\mathcal{G}:|\widehat{\beta}_j| \neq 0\right\}.
\end{equation*}
The set $\Aset$ is clearly dependent on the data and the values $\gamma$ but we omit this details for ease of notation. Calculation of the degrees of freedom involves the computation of the trace of $\partial \widehat{\by}_\gamma / \partial \by$. While for Lasso this result simplifies to the size of the active set \citep{zou2007}, for its adaptive version the presence of data-dependent weights makes its calculation less straightforward. We are now ready to present the first  contribution of our paper, which is an unbiased estimator of the degrees of freedom for Adaptive Lasso model in Equation \eqref{eq:convex_regularized_problem_adalasso} for general expression of the adaptive weights.
\begin{theorem}
\label{th:df_adalasso_general_weights}
Let $\betahat$ the solution to the Adaptive Lasso problem in Equation \eqref{eq:convex_regularized_problem_adalasso} with weights  $w_j = w_j(|\widehat{\beta}_j^\mathsf{LS}|)$ and $\gamma \in (\gamma_l, \gamma_{l+1})$. Denote with $\Aset$ the corresponding active set, and with $\bX_{\Aset}$ the design matrix restricted to the active set. 
Define a mapping $\pi: \Aset \rightarrow\{1,2, \ldots,|\Aset|\}$ such that for each $j \in \Aset$, $\pi(j)=i$ if $\Aset_i=j$.
An unbiased estimate
of the degrees of freedom is
\begin{equation}
\widehat{df}_\gamma =  |\Aset|-\gamma \sum_{j \in \Aset}  \frac{\partial w_j(z)}{\partial z} \bigg|_{z = \widehat{\beta}_j^\mathsf{LS}},
\end{equation}
for the orthonormal design and
\begin{equation}
    \widehat{df}_\gamma =|\Aset|
    -\gamma \sum_{j \in \Aset} \sgn(\widehat{\beta}_j) \sgn(\widehat{\beta}_j^\mathsf{LS}) \frac{\partial w_j(z)}{\partial z} \bigg|_{z = \widehat{\beta}_j^\mathsf{LS}} 
[(\bX_{\Aset}^{\T} \bX_{\Aset} )^{-1}]_{\pi(j),\pi(j)},
\end{equation}
for non-orthonormal designs.
\end{theorem}
\noindent We now focus on the popular choice $w_j=1/|\widehat{\beta}^{\mathsf{LS}}_j|$ for $j=1,\dots,G$ and specify the previous result for both the orthonormal and non-orthonormal cases. Analogous results for the less common case $w_j=\exp(-\alpha |\widehat{\beta}_j^\mathsf{LS}|)$ are presented in 
the Supplementary Materials.
\begin{corollary}
\label{th:df_adalasso_ortho}
    Let $\bX^{\T}\bX=\bI_p$ and $w_j = 1/|\widehat{\beta}_j^{\mathsf{LS}}|$. Under the settings of Theorem \ref{th:df_adalasso_general_weights} an unbiased estimate of the degrees of freedom is
    \begin{equation}
    \label{eq:dof_adalasso_ortho}
        \widehat{df}_\gamma = \vert\Aset\vert+\gamma\sum_{j\in\Aset} 
        \frac{1}{(\widehat{\beta}_j^{\mathsf{LS}})^2}.
    \end{equation}
\end{corollary}
\begin{corollary}
\label{th:df_adalasso_nonortho}
    Let $\bX^{\T}\bX\neq \bI_p$ and $w_j = 1/|\widehat{\beta}_j^{\mathsf{LS}}|$. Under the settings of Theorem \ref{th:df_adalasso_general_weights} an unbiased estimate of the degrees of freedom is
\begin{equation} 
\label{eq:dof_adalasso_nonortho}
\widehat{df}_\gamma=\vert\Aset\vert+\gamma\sum_{j\in\Aset} \sgn(\widehat{\beta}_{j}) \frac{ \sgn(\widehat{\beta}_j^{\mathsf{LS}})}{(\widehat{\beta}_j^{\mathsf{LS}})^2} [(\bX_{\Aset}^{\T} \bX_{\Aset} )^{-1}]_{\pi(j),\pi(j)}.
\end{equation}
\end{corollary}
\noindent Equations 
\eqref{eq:dof_adalasso_ortho} and \eqref{eq:dof_adalasso_nonortho} reveal that the degrees of freedom extend beyond simply counting the active set size. Specifically, they include an additional term that depends on the weights' choice and on the regularization parameter $\gamma$. For non-orthonormal designs, the sign of the Adaptive Lasso estimate $\betahat$, the sign and magnitude of the least squares estimate $\betahat^{\mathsf{LS}}$ and the diagonal elements of the active set's precision matrix must also be incorporated. 
Notably, the degrees of freedom expression exhibits a piecewise linear structure. This result stands in sharp contrast to the Lasso, where the estimated degrees of freedom form a piecewise constant function of $\gamma$. The Lasso’s piecewise constant degrees of freedom, however, can be recovered as a special case of the Adaptive Lasso by choosing fixed weights $w_j$. In this scenario, the gradient $\partial w_j(\by)/\partial \by$ vanishes for all $j \in \mathcal{G}$, reducing $\widehat{{df}}_\gamma$ to a piecewise constant function, consistent with classical Lasso results. Under certain conditions, $\widehat{{df}}_\gamma \geq |\mathcal{A}|$, demonstrating that the common practice of using only the active set size systematically underestimates the true degrees of freedom, and this relationship displays strictly positive but successively diminishing slopes. This is formalized in the following corollary.
\begin{corollary}
\label{coro:linearincreasing}
Let $w_j(|\widehat{\beta}_j^\mathsf{LS}|)$ be decreasing functions of $|\widehat{\beta}_j^\mathsf{LS}|$ for all $j\in \cal G$. Denote with $\Aset$ the corresponding active set. If $\sgn{(\widehat{\beta}_j)} \cdot \sgn{(\widehat{\beta}_j^\mathsf{LS})} \geq 0,$  for all $j\in \cal G$ then, for $\gamma \in (\gamma_l, \gamma_{l+1})$, 
$\widehat{df}_\gamma = |\Aset| + b_l \gamma$ with $b_l >0$ for any $l = 1, \dots, L-1$, leading to $\widehat{df}_\gamma > |\Aset|$. If the size of active set is a decreasing function of $\gamma$, then $b_l > b_{l+1}$ for any $l=1,\ldots, L-1$.
\end{corollary}
\noindent The conditions of Corollary~\ref{coro:linearincreasing} are typically satisfied in practice, as standard weights are usually monotonically decreasing functions of the absolute value of least squares estimates. The remaining two conditions, pertaining to the signs of the estimates and the size of the active set, are also frequently met in general settings and are guaranteed to hold for orthogonal designs.
A counterintuitive property emerges from Corollary~\ref{coro:linearincreasing}: within adjacent intervals defined by the transition points the estimated degrees of freedom $\widehat{df}_\gamma$ may be higher for $\gamma \in (\gamma_l, \gamma_{l+1})$ than for $\gamma \in (\gamma_{l-1}, \gamma_l)$ even if the active set size is actually decreased. This implies that a larger $\gamma$ (and thus stronger regularization) does not always correspond to a simpler model, at least in terms of \emph{effective complexity}, in line with the results of \cite{kaufman_rosset.2014} and \cite{janson_etal.2015}. This apparent inconsistency arises because, as already discussed, $\widehat{df}_\gamma$ depends not only on the number of active parameters but also on their adaptive weighting structure. When $\gamma$ crosses a transition point $\gamma_l$, the adaptive weights---which themselves depend on $\by$---induce nonlinear interactions between the retained coefficients and the regularization path. These interactions can amplify the sensitivity of the fitted model to the data, effectively ``consuming more information'' despite a sparser active set.  
\subsection{The degrees of freedom of Group Lasso}
\label{sec:glasso_given_weights}
Let $\bX$, $\by$, $\bbeta$ and $\gamma$ as in the previous section. Let $G \leq p \in \mathbb{N}$, be the number of groups and $\bw \in \mathbb{R}^G_+$ a vector of non-negative weights associated to each group. The Group Lasso \citep{yuan:2006:glasso} estimator is defined as
\begin{equation}
\label{eq:convex_regularized_problem_glasso}
\widehat{\boldsymbol{\beta}}=\arg\min_{\bbeta} \frac{1}{2}\Vert\by-\bX \bbeta \Vert_2^2+\gamma \sum_{g=1}^Gw_g\Vert \bbeta_g\Vert_2.
\end{equation}
For this problem the group weights $w_g$ are assumed to be independent from the response, a common choice being setting $w_g \propto n_g$  where $n_g$ is the size of the corresponding group. If the weights are determined in an \textit{adaptive} way, then the model is known as Adaptive Group Lasso \citep{wangleng2008csda} for which we provide results in the following Section \ref{sec:glasso_adaptive_weights}.
For the Group Lasso, let $\mathcal{G}_p=\{1,\dots,p\}$ be the index set of all variables and $\mathcal{G}_G=\{1,\dots,G\}$ the index set of all groups. Then, define the index set of active variables and the index set of active groups as
\begin{equation*}
\label{eq:dof_glasso_Aset_def}
\Aset_p=\left\{j\in\mathcal{G}_p:|\widehat{\beta}_j| \neq 0\right\}, \quad \Aset_G=\left\{g\in\mathcal{G}_G:\lVert\widehat{\bbeta}_g\rVert_2 \neq 0\right\}.
\end{equation*}
Again, we omit the hat from $\Aset_p$ and $\Aset_G$, as well as their dependence from $\gamma$. We are now ready to state the main result about the computation of degrees of freedom of Group Lasso under the orthonormal design.
\begin{theorem}
\label{th:df_glasso_ortho}
Let $\bX^{\T}\bX=\bI_p$, $\betahat$ be the solution to the Group Lasso problem in Equation \eqref{eq:convex_regularized_problem_glasso} with groups cardinality equal to $n_1, \ldots, n_G,$ 
and $\gamma \in (\gamma_l, \gamma_{l+1})$. Denote with $\Aset_G$ the corresponding active group set. An unbiased estimate of the degrees of freedom is 
\begin{equation}
\label{eq:dof_glasso_ortho}
\widehat{df}_\gamma = \trace \big[ \big(\bI_n+\gamma \bB \big)^{-1} \bA  \big],
\end{equation} 
where $\bA = \bX_{\Aset_G}\bX_{\Aset_G}^{\T},$ $\bB = \bX_{\Aset_G}\bPi_{\Aset_G}\bX_{\Aset_G}^{\T}$, and 
\begin{equation}
\label{eq:ABPi_g_definition} 
\bPi_{\Aset_G} = \mathrm{blockdiag}\bigg(w_g\bigg[\frac{\bI_{n_g}}{\Vert\betahat_g\Vert_2}-\frac{\betahat_g\betahat_g^{\T}}{\Vert\betahat_g\Vert_2^3}\bigg]\bigg).
\end{equation}
\end{theorem}
\noindent We now consider the more general case where the design matrix is not orthonormal.
\begin{theorem}
\label{th:df_glasso_nonortho} 
Let  $\bX^{\T}\bX \neq\bI_p$, $\betahat$ the solution to the Group Lasso problem in Equation \eqref{eq:convex_regularized_problem_glasso} with groups cardinality equal to $n_1, \ldots, n_G$ 
and $\gamma \in (\gamma_l, \gamma_{l+1})$. Denote with $\Aset_G$ the corresponding active group set. An unbiased estimate of the degrees of freedom is  
\begin{equation}
\label{eq:dof_glasso_nonortho}
\widehat{df}_\gamma = \trace \big[ \big(\bI_n+\gamma \bB\big)^{-1} \bA  \big],
\end{equation}
where $\bA = \bX_{\Aset_G}(\bX_{\Aset_G}^{\T}\bX_{\Aset_G})^{-1}\bX_{\Aset_G}^{\T}$,
$\bB = \bX_{\Aset_G}(\bX_{\Aset_G}^{\T}\bX_{\Aset_G})^{-1} \bPi_{\Aset_G} (\bX_{\Aset_G}^{\T}\bX_{\Aset_G})^{-1} \bX_{\Aset_G}^{\T}$,
and
\begin{equation*}
\bPi_{\Aset_G} = \mathrm{blockdiag}\bigg(w_g\bigg[\frac{\bI_{n_g}}{\Vert\betahat_g\Vert_2}-\frac{\betahat_g\betahat_g^{\T}}{\Vert\betahat_g\Vert_2^3}\bigg]\bigg).
\end{equation*}
\end{theorem}
\noindent The results in Equations \eqref{eq:dof_glasso_ortho} and \eqref{eq:dof_glasso_nonortho} are apparently different from those of the Lasso and Adaptive Lasso. 
Furthermore, Equation \eqref{eq:dof_glasso_ortho} looks different from that reported in the original Group Lasso paper by \cite{yuan:2006:glasso} and from a related expression reported by \cite{vaiter_etal.2017}. The relations between these expressions are further discussed in 
in the Supplementary Materials. 
From Equations \eqref{eq:dof_glasso_ortho} and \eqref{eq:dof_glasso_nonortho}, it becomes evident that an unbiased estimate of the degrees of freedom for Group Lasso is again not merely the trace of the projection matrix---which would equate to the number of active variables---but rather the trace of a slightly different matrix. This matrix additionally depends on the design matrix and the vector of estimated coefficients. Consistently, using the size of the active set as an estimator of the degrees of freedom is generally not valid for Group Lasso. The following corollary provides bounds for the degrees of freedom obtained in Theorems \ref{th:df_glasso_ortho} and \ref{th:df_glasso_nonortho}. 
\begin{corollary}
\label{coro:bounds_df_glasso}
   The estimated degrees of freedom of Group Lasso given by Theorem \ref{th:df_glasso_ortho} and \ref{th:df_glasso_nonortho} are a continuous function of $\gamma$, for $\gamma  \in (\gamma_l, \gamma_{l+1})$ and
   $$ |\Aset_G| \ \leq \ \widehat{df}_\gamma \ \leq \ |\Aset_p|,$$
   where we have $|\Aset_G| = \widehat{df}_\gamma = |\Aset_p|=0$ if there are no active groups, and $\widehat{df}_\gamma = |\Aset_p|=p$ in the limiting case $\gamma=0$ of no penalization. 
\end{corollary}
\noindent Differently from Lasso, in Group Lasso the Stein's estimator of degrees of freedom lead to a non-constant piecewise continuous function. This means that by replacing the $\ell_1$ norm with the $\ell_2$ norm we lose the perfect balance, typical of the Lasso, between using the response for variable selection and applying shrinkage. Specifically, the estimated degrees of freedom of Group Lasso are always lower than $|\Aset_p|$, being equal to such value only in limiting cases. Thus the use of the $\ell_2$ norm results in a more parsimonious usage of the information available in data. 
As expected, we recover the Lasso's expression if we set the number of elements in each group to one, as formally discussed in 
the Supplementary Materials.
A natural question is whether the estimated degrees of freedom of the Group Lasso decrease as $\gamma$ increases. The answer is formalized in the next two theorems.
\begin{theorem}
\label{th:monotonicity_df_glasso}
If the matrix $\bPi_{\Aset_G} + \gamma \mathrm{d} \bPi_{\Aset_G} / \mathrm{d}\gamma $ is positive definite, then the estimated degrees of freedom of Group Lasso given by Theorem \ref{th:df_glasso_ortho} and \ref{th:df_glasso_nonortho} are non-increasing function of $\gamma$, for $\gamma \in (\gamma_l, \gamma_{l+1})$.
\end{theorem}
\begin{theorem}
\label{th:matrix_Pi_indefinite}
    If $\bX^{\T} \bX = \bI_p$ or $\betahat_{\Aset_G}$ is an eigenvector of the matrix $(\bX_{\Aset_G}^{\T} \bX_{\Aset_G} + \gamma \bPi_{\Aset_G})^{-1} \bW_{\Aset_G}$, the matrix $\bPi_{\Aset_G} + \gamma \mathrm{d} \bPi_{\Aset_G} / \mathrm{d}\gamma$ is positive semidefinite. Otherwise, it is indefinite.
\end{theorem}
\begin{corollary}
\label{coro:sufficient_nonortho}
    For a general design matrix, monotonicity of the estimated degrees of freedom cannot be established by positive semidefiniteness of the matrix $\bPi_{\Aset_G} + \gamma \mathrm{d} \bPi_{\Aset_G} / \mathrm{d}\gamma$, in general. Nevertheless
    , a sufficient condition is $\trace(\bA\bB) > 0$, where $\bA = \bPi_{\Aset_G} + \gamma \mathrm{d} \bPi_{\Aset_G} / \mathrm{d}\gamma$,  $\bB = \bM_{\Aset_G}^{\T}\left( \bI_n+\gamma  \bB_\gamma\right)^{-1}  \bA\left( \bI_n+\gamma  \bB_\gamma\right)^{-1}\bM_{\Aset_G}$, $\bM_{\Aset_G}=\bX_{\Aset_G}(\bX_{\Aset_G}^{\T}\bX_{\Aset_G})^{-1}$.
\end{corollary}
\noindent The contribution of Theorems \ref{th:monotonicity_df_glasso} and \ref{th:matrix_Pi_indefinite} is the characterization of the gradient of the estimated degrees of freedom of the Group Lasso in terms of the spectral decomposition of the matrix $\bPi_{\Aset_G}+\gamma\mathrm{d}\bPi_{\Aset_G}/\mathrm{d}\gamma$. This spectral perspective provides a fine-grained understanding of the underlying geometric and algebraic mechanisms that govern the behavior of the estimator as the regularization parameter $\gamma$ varies. Consistently, the results presented in such theorems, and the associated corollaries in the Supplementary Materials, constitute a substantial advancement in understanding the estimated degrees of freedom of the Group Lasso.
In particular, these results go significantly beyond the existing literature by providing an explicit and general expression for the derivative of the estimated degrees of freedom with respect to $\gamma$, valid under general non-orthogonal design conditions. Previous studies typically restrict attention to the orthonormal design case, under which significant simplifications arise. Specifically, only in the orthonormal setting the matrix $\bPi_{\Aset_G}+\gamma\mathrm{d}\bPi_{\Aset_G}/\mathrm{d}\gamma$ is positive semidefinite, as stated in Theorem \ref{th:matrix_Pi_indefinite}. In this case, negative eigenvalues vanish, while the other eigenvalues reduce to simple closed-form expressions. Moreover, in this settings the matrices $\bPi_g$, as well as vectors $\betahat_g$ and their norms, scale linearly with $\gamma$ for each active group $g \in \Aset_G$, as established 
in the Supplementary Materials, indicating a smooth shrinkage along each of the directions $\betahat_g$. Notably, this behavior no longer holds in the general non-orthogonal design case.
In fact, in the general non-orthogonal case, the analysis rigorously reveals that the matrix $\bPi_{\Aset_G}+\gamma\mathrm{d}\bPi_{\Aset_G}/\mathrm{d}\gamma$ becomes indefinite, and monotonicity of the estimated degrees of freedom should be established by other sufficient conditions, as outlined for instance in Corollary \ref{coro:sufficient_nonortho}. 
\subsection{The degrees of freedom of Adaptive Group Lasso}
\label{sec:glasso_adaptive_weights}

Let $\bX$, $\by$, $\bbeta$, $\gamma$, and $G$ as in the previous sections. Let  $\bw(\by) \in \mathbb{R}^G_+$ a vector of non-negative data-dependent weights associated to each group. The Adaptive Group Lasso estimator is defined as
\begin{equation}
\label{eq:convex_regularized_problem_adaglasso}
\widehat{\boldsymbol{\beta}}=\arg\min_{\bbeta} \frac{1}{2}\Vert\by-\bX \bbeta \Vert_2^2+\gamma \sum_{g=1}^Gw_g(\by)\Vert \bbeta_g\Vert_2.
\end{equation}
As in previous section, let $\Aset_p$ and $\Aset_G$ bet the index sets of active groups and active variables, respectively.
We are now ready to present the unbiased estimator for the degrees of freedom in Adaptive Group Lasso model in Equation \eqref{eq:convex_regularized_problem_adaglasso} under  general expression for the adaptive weights.
\begin{theorem}
\label{th:df_adaglasso_general_weights}
Let $\betahat$ be the solution to the Adaptive Group Lasso problem in Equation \eqref{eq:convex_regularized_problem_adaglasso} with groups cardinality equal to $n_1, \ldots, n_G,$ weights defined as $w_g = w_g(\lVert \widehat{\bbeta}^\mathsf{LS} \rVert_2)$ and $\gamma \in (\gamma_l, \gamma_{l+1})$. Denote with $\Aset_G$ the corresponding active group set, and with $\bX_{\Aset_G}$ the design matrix restricted to the active group set. An unbiased estimate
of the degrees of freedom is
\begin{equation}
\label{eq:dof_adaglasso_general_weights}
\widehat{df}_\gamma = \trace\big[\big(\bI_n + \gamma \bB \big)^{-1} \big(\bA - \gamma \bC \big)\big],
\end{equation}
where $\bA = \bX_{\Aset_G} \bX_{\Aset_G}^{\T}$, $\bB = \bX_{\Aset_G} \bPi_{\Aset_G} \bX_{\Aset_G}^{\T}$, $\bC = \bX_{\Aset_G} \bPhi_{\Aset_G} \bX_{\Aset_G}^{\T}$ for the orthonormal design, $\bA = \bX_{\Aset_G}(\bX_{\Aset_G}^{\T} \bX_{\Aset_G})^{-1}\bX_{\Aset_G}^{\T}$, $\bB = \bX_{\Aset_G}(\bX_{\Aset_G}^{\T} \bX_{\Aset_G})^{-1} \bPi_{\Aset_G}(\bX_{\Aset_G}^{\T} \bX_{\Aset_G})^{-1}\bX_{\Aset_G}^{\T}$, \\
$\bC = \bX_{\Aset_G} (\bX_{\Aset_G}^{\T} \bX_{\Aset_G})^{-1}\bPhi_{\Aset_G} (\bX_{\Aset_G}^{\T} \bX_{\Aset_G})^{-1}\bX_{\Aset_G}^{\T}$ for non-orthonormal designs, and
\begin{equation*}
   \bPi_{\Aset_G} = \mathrm{blockdiag}\bigg( w_g \bigg[ \frac{\bI_{n_g}}{\lVert \widehat{\bbeta}_g \rVert_2} - \frac{\widehat{\bbeta}_g \widehat{\bbeta}_g^{\T}}{\lVert \widehat{\bbeta}_g \rVert_2^3} \bigg] \bigg), \quad \bPhi_{\Aset_G} = \mathrm{blockdiag}\bigg(  \frac{\widehat{\bbeta}_g}{\lVert \widehat{\bbeta}_g \rVert_2 } \frac{\partial w_g(\lVert \widehat{\bbeta}^\mathsf{LS}\rVert_2)}{\partial \lVert\widehat{\bbeta}^\mathsf{LS} \rVert_2}\frac{(\widehat{\bbeta}^\mathsf{LS})^{\T}}{\lVert \widehat{\bbeta}^\mathsf{LS} \rVert_2}\bigg). 
\end{equation*}
\end{theorem}
\noindent Notably, these expressions incorporate two competing effects: an inflation effect caused by the use of adaptive weights, consistently with Section \ref{sec:adalasso},  and a contraction effect induced by the $\ell_2$ norm, consistently with Section \ref{sec:glasso_given_weights}. The relative strength of the two parts cannot be generally determined, and wether the trace of this quantity is greater or lower than $|\Aset_p|$ is not known.
We now specify the expression for the degrees of freedom of Group Lasso with adaptive weights for both orthonormal and non-orthonormal designs, under the typical assumption for the  weights $w_g = 1/\lVert \widehat{\bbeta}^{\mathsf{LS}}_g \rVert_2$.
\begin{corollary}
\label{th:df_adaglasso_ortho}
    Let $\bX^{\T}\bX=\bI_p$ and $w_g = 1/\lVert \widehat{\bbeta}_g^{\mathsf{LS}}\rVert_2$. Under the settings of Theorem \ref{th:df_adaglasso_general_weights} an unbiased estimate of the degrees of freedom is
\begin{equation}
\label{eq:dof_adaglasso_ortho}
\widehat{df}_\gamma = \trace \big[ \big(\bI_n+\gamma \bB \big)^{-1} \big(\bA + \gamma\bC \big) \big],
\end{equation}
where $\bA = \bX_{\Aset_G}\bX_{\Aset_G}^{\T}$, $\bB =\bX_{\Aset_G}\bPi_{\Aset_G}\bX_{\Aset_G}^{\T}$, $\bC =  \bX_{\Aset_G} \bPhi_{\Aset_G} \bX_{\Aset_G}^{\T},$ and
\begin{equation*}
    \bPi_{\Aset_G} = \mathrm{blockdiag}\bigg(w_g\bigg[\frac{\bI_{n_g}}{\Vert\betahat_g\Vert_2}-\frac{\betahat_g\betahat_g^{\T}}{\Vert\betahat_g\Vert_2^3}\bigg]\bigg),\quad \bPhi_{\Aset_G} = \mathrm{blockdiag} \bigg( \frac{\widehat{\bbeta}_g}{\lVert \widehat{\bbeta}_g \rVert_2}\frac{(\widehat{\bbeta}_{g}^{\mathsf{LS}})^{\T}}{\lVert \widehat{\bbeta}_{g}^{\mathsf{LS}}\rVert_2^3} \bigg).
\end{equation*}
\end{corollary}
\begin{corollary}
\label{th:df_adaglasso_nonortho}
Let $\bX^{\T}\bX\neq \bI_p$ and $w_g = 1/\lVert \widehat{\bbeta}_g^{\mathsf{LS}}\rVert_2$. Under the settings of Theorem \ref{th:df_adaglasso_general_weights} an unbiased estimate of the degrees of freedom is
\begin{equation}
\widehat{df}_\gamma = \trace \big[ \big(\bI_n+\gamma \bB \big)^{-1} \big(\bA + \gamma\bC \big) \big],
\end{equation}
where $\bA = \bX_{\Aset_G} (\bX_{\Aset_G}^{\T}\bX_{\Aset_G})^{-1}\bX_{\Aset_G}^{\T}$, $\bB =\bX_{\Aset_G}(\bX_{\Aset_G}^{\T}\bX_{\Aset_G})^{-1}\bPi_{\Aset_G}(\bX_{\Aset_G}^{\T}\bX_{\Aset_G})^{-1}\bX_{\Aset_G}^{\T}$, \\
$\bC =  \bX_{\Aset_G} (\bX_{\Aset_G}^{\T}\bX_{\Aset_G})^{-1}\bPhi_{\Aset_G} (\bX_{\Aset_G}^{\T}\bX_{\Aset_G})^{-1}\bX_{\Aset_G}^{\T},$ and
\begin{equation*}
   \bPi_{\Aset_G} = \mathrm{blockdiag}\bigg(w_g\bigg[\frac{\bI_{n_g}}{\Vert\betahat_g\Vert_2}-\frac{\betahat_g\betahat_g^{\T}}{\Vert\betahat_g\Vert_2^3}\bigg]\bigg), \quad
   \bPhi_{\Aset_G} = \mathrm{blockdiag} \bigg( \frac{\widehat{\bbeta}_g}{\lVert \widehat{\bbeta}_g \rVert_2}\frac{(\widehat{\bbeta}_{g}^{\mathsf{LS}})^{\T}}{\lVert \widehat{\bbeta}_{g}^{\mathsf{LS}}\rVert_2^3} \bigg). 
\end{equation*}
\end{corollary}
\noindent Similarly to Corollary \ref{coro:linearincreasing} we can compare the degrees of freedom in Equation \eqref{eq:dof_adaglasso_general_weights} with those of the Group Lasso. Specifically, we show that under certain conditions the Adaptive Group Lasso degrees of freedom estimates are no lower than those of the Group Lasso. This is formalized in the following Corollary.
\begin{corollary}
\label{coro:bounds_df_adaglasso}
Let $w_g(\lVert \widehat{\bbeta}_g \rVert_2)$ be decreasing functions of $\lVert \widehat{\bbeta}_g \rVert_2$ for all $g\in \mathcal{G}_G$. If $\widehat{\bbeta}^{\T}_g \widehat{\bbeta}^{\mathsf{LS}}_g  \geq 0,$  for all $g\in \mathcal{G}_G$ then, for $\gamma \in (\gamma_l, \gamma_{l+1})$,
$\widehat{df}_\gamma$ for Adaptive Group Lasso, as defined in Equation \eqref{eq:dof_adaglasso_general_weights}, are bounded below by  $\widehat{df}_\gamma$ for Group Lasso, as defined in Equation \eqref{eq:dof_glasso_ortho} and \eqref{eq:dof_glasso_nonortho}, for any $l = 1, \dots, L-1$.
\end{corollary}
\noindent In the orthonormal design setting, an alternative expression for the estimated degrees of freedom appearing in Equation \eqref{eq:dof_adaglasso_ortho} in the spirit of \cite{yuan:2006:glasso} can be derived, and it is presented 
in the Supplementary Materials. In the same document, 
we show that Lasso and Group Lasso are special cases of Adaptive Group Lasso, and their estimated degrees of freedom can be recovered from Theorem \ref{th:df_adaglasso_general_weights}.
\section{Empirical Assessment}
\label{sec:exper}

In this section, we conduct two numerical experiments to validate and contextualize our theoretical findings. First, we generate synthetic data to empirically demonstrate the consistency of our theoretical results. Second, we analyze a real-world dataset to evaluate the utility of the estimated degrees of freedom in model selection. Notably, the use of our results  achieves performance comparable to computationally intensive cross-validation while diverging significantly from the naive practice of equating the degrees of freedom to the size of the active set.
\subsection{Synthetic data}
\label{subsec:syntheticdata}

We let  $n=100$ and $p=30$ and generate a $n \times p$ data matrix $\bX = (\bX_1, \ldots, \bX_p)$ where $\bX_j \sim \text{N}(0,\bI_n)$ for $j=1, \ldots, p$. We fix $\bbeta = (5, -5, 5, 3, -3, 3, 1, -1, 1, 0, \ldots, 0)$ thus having the last 21 coefficients equal to zero, and consequently let $\by^\star=\bX \bbeta$. For $B=10{,}000$ replications we perturb $\by^\star$ with independent noise generated from a Gaussian distribution with zero mean and variance equal to a value leading to a signal-to-noise ratio equal to 4, obtaining $\by_b = \bX \bbeta + \bvarepsilon_b$. Estimation of  $\bbeta$ is obtained with the estimators presented in Equations \eqref{eq:convex_regularized_problem_adalasso}, \eqref{eq:convex_regularized_problem_glasso}, and \eqref{eq:convex_regularized_problem_adaglasso}. For the Group Lasso and its adaptive version, we grouped three consecutive coefficients leading to $n_g=3$ for $g=1, \dots, G$ and $G=10$. For the Adaptive Lasso, weights are equal to the reciprocal of the absolute values of least squares estimates of the parameters. For the Adaptive Group lasso the weight of each group equals  the reciprocal of the $\ell_2$ norm of corresponding least squares estimates. 
\begin{figure*}[t]
    \centering
    \includegraphics[width=0.9\linewidth]{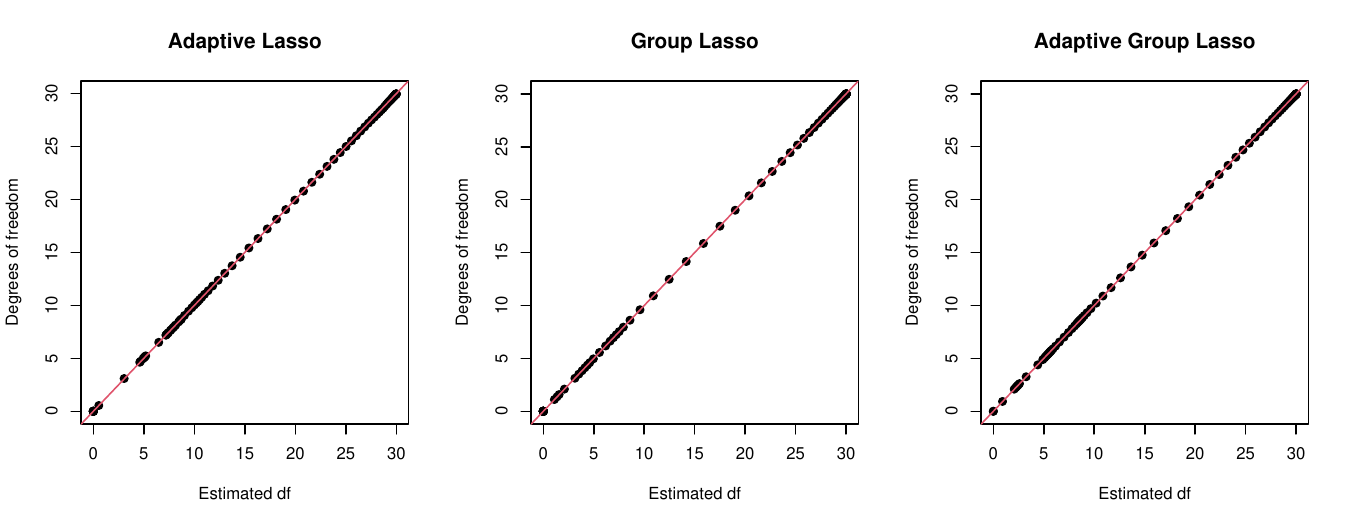}
    \caption{Estimated degrees of freedom using the appropriate theorem (x axis) versus degrees of freedom computed with the general covariance formula \eqref{eq:df_general_expression} for Adaptive Lasso, Group Lasso and Adaptive Group Lasso.}
    \label{fig:convergence}
\end{figure*}
For each method, an appropriate range for the regularization parameter is considered, and the entire path of solutions of the corresponding optimization problem is computed. For each replicate $b=1, \dots, B$, we compute the estimated degrees of freedom  using results of Theorems \ref{th:df_adalasso_ortho}, \ref{th:df_glasso_ortho}, and \ref{th:df_adaglasso_ortho}, respectively. The unbiasedness of these estimators is verified by comparing their empirical average 
with the theoretical degrees of freedom computed via the covariance formula in Equation \eqref{eq:df_general_expression}. Figure \ref{fig:convergence} demonstrates strong agreement across all methods. Note that while Equation \eqref{eq:df_general_expression} serves as a benchmark, it cannot be applied in practice as it requires knowledge of the true parameter values.
\begin{table*}[h!]
\begin{center}
\caption{Empirical distribution of variables entering the model over 10000 repetitions (the true value is 9) for Adaptive Lasso, Group Lasso, and Adaptive Group Lasso.}
\label{tab:ncomp}
\begin{tabular}{l c c c c c c c}
\hline
 & $\leq7$ & 8 & \textbf{9} & 10 & 11 & 12 & $\geq13$ \\
 \hline
 Adaptive Lasso & 341 & 1154 & 4385 & 2217 & 1042 & 461 & 400 \\  
 Group Lasso & 0 & 0 & 3786 & 0 & 0 & 3358 & 2856 \\ 
 Adaptive Group Lasso & 9 & 0 & 9449 & 0 & 0 & 504 & 38 \\
 \hline
\end{tabular}
\end{center}
\end{table*}
Additionally, for each replicate, we determine the optimal regularization parameter $\gamma$ with respect to the BIC by minimizing the expression in Equation \eqref{eq:BIC}. Here, the degrees of freedom are computed using the method corresponding to each regularization technique, and  $\sigma^2$ replaced by its estimate obtained from ordinary linear regression.
The distribution of the number of retained coefficients across the $B$ replicates, corresponding to each optimal $\gamma$, is presented in Table \ref{tab:ncomp}.
The results indicate that Adaptive Lasso and Adaptive Group Lasso select 9 variables in the majority of cases, whereas Group Lasso frequently retains more than 9. This discrepancy arises because Group Lasso, lacking an adaptive mechanism, often fails to exclude noise-related coefficients.

\subsection{Diabetes data}
\label{subsec:realdata}
\begin{figure*}[!ht]
    \centering
    \includegraphics[width=0.64\linewidth]{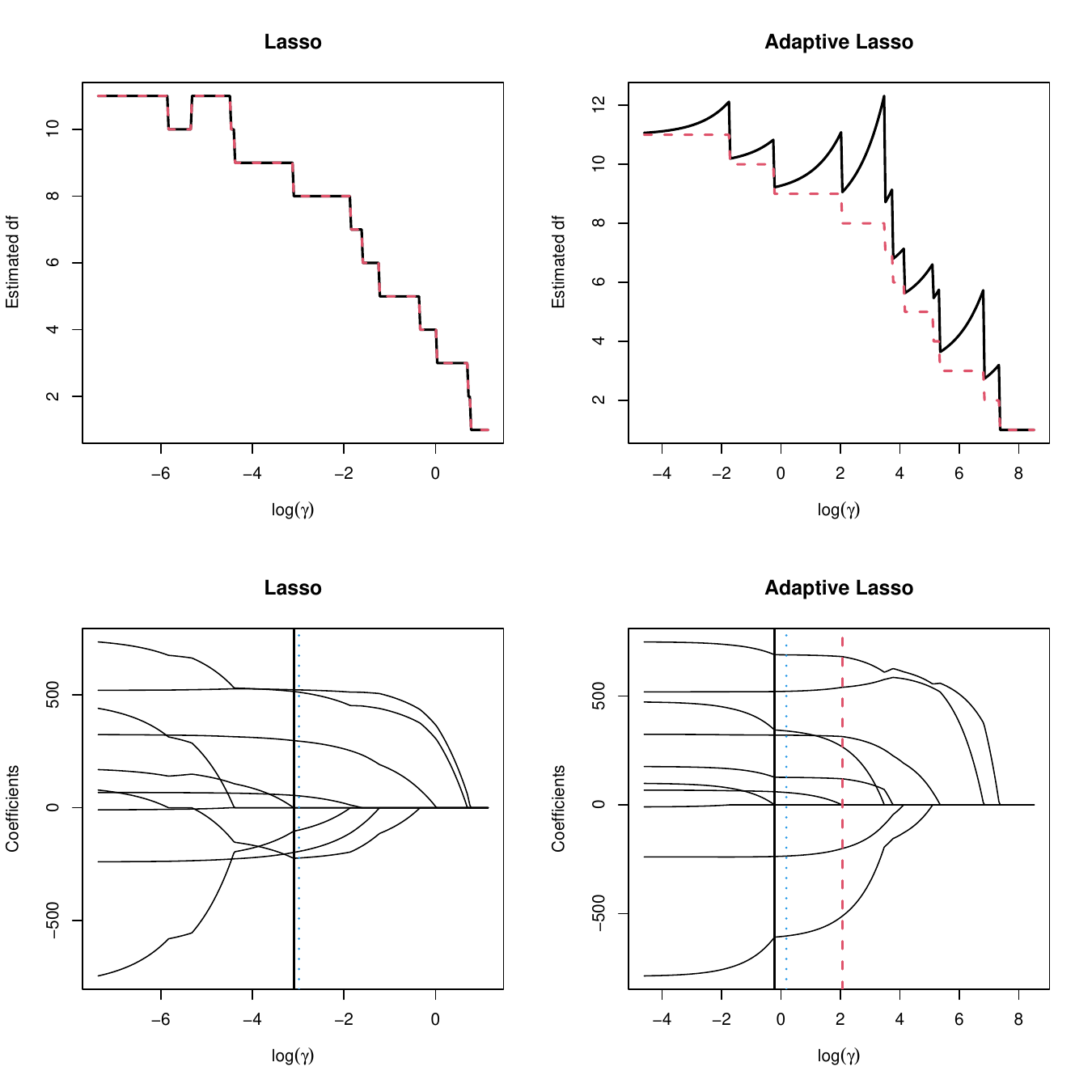}
    \caption{Results of the Lasso (left) and Adaptive Lasso (right) estimation for the Diabetes (small) data. Estimated degrees of freedom (upper panels, continuous lines) along with the size of the active set (dashed lines). Complete solution path (lower panels) with vertical lines denoting the best $\gamma$ according to correct BIC (continuous line), cross validation (dotted line), BIC with active set size as estimator of the degrees of freedom (dashed).}
    \label{fig:diabetes_small}
\end{figure*}
To illustrate  the utility of our findings we analyze the Diabetes dataset \citep{efron2004least,zou2007}. This dataset is available in two versions:  a \textit{small} version with $p=10$ covariates, and a \textit{large} version with $p=64$ covariates. The number of observations is $n=442$. The covariates are individual characteristics and biomarkers while the response variable is a quantitative measure of disease progression. 
We use this dataset to illustrate our results with two separate analyses. First, we use the original $p=10$ covariates of the \textit{small} version and fit a linear model with Lasso and Adaptive Lasso penalties discussing model selection via BIC. In a second analysis we discretize the continuous covariates into ordinal categorical variables for ``high'', ``medium-high'', ``medium-low'', and ``low'' levels of each biomarkers and encode these categorical variable into three dummy variables for each biomarker. After defining groups of dummy variable associated to the same categorical variable, we fit linear regression models with Group Lasso and Adaptive Group Lasso penalties. 
We start by discussing the results for the Lasso and Adaptive Lasso. The upper panels of Figure~\ref{fig:diabetes_small} report the estimated degrees of freedom as a function of $\log(\gamma)$. For comparison, we include the size of the active set (shown as a dashed line) to highlight discrepancies arising from the common but erroneous practice of equating this quantity with the estimated degrees of freedom across all penalization methods. It can be seen that for the Adaptive lasso, the correct estimated degrees of freedom (solid curve) are strictly lower-bounded by the active set size. This reveals that using the active set size as a proxy for degrees of freedom systematically underestimates the true complexity of the model, leading to 
\begin{figure*}[!ht]
    \centering
    \includegraphics[width=0.64\linewidth]{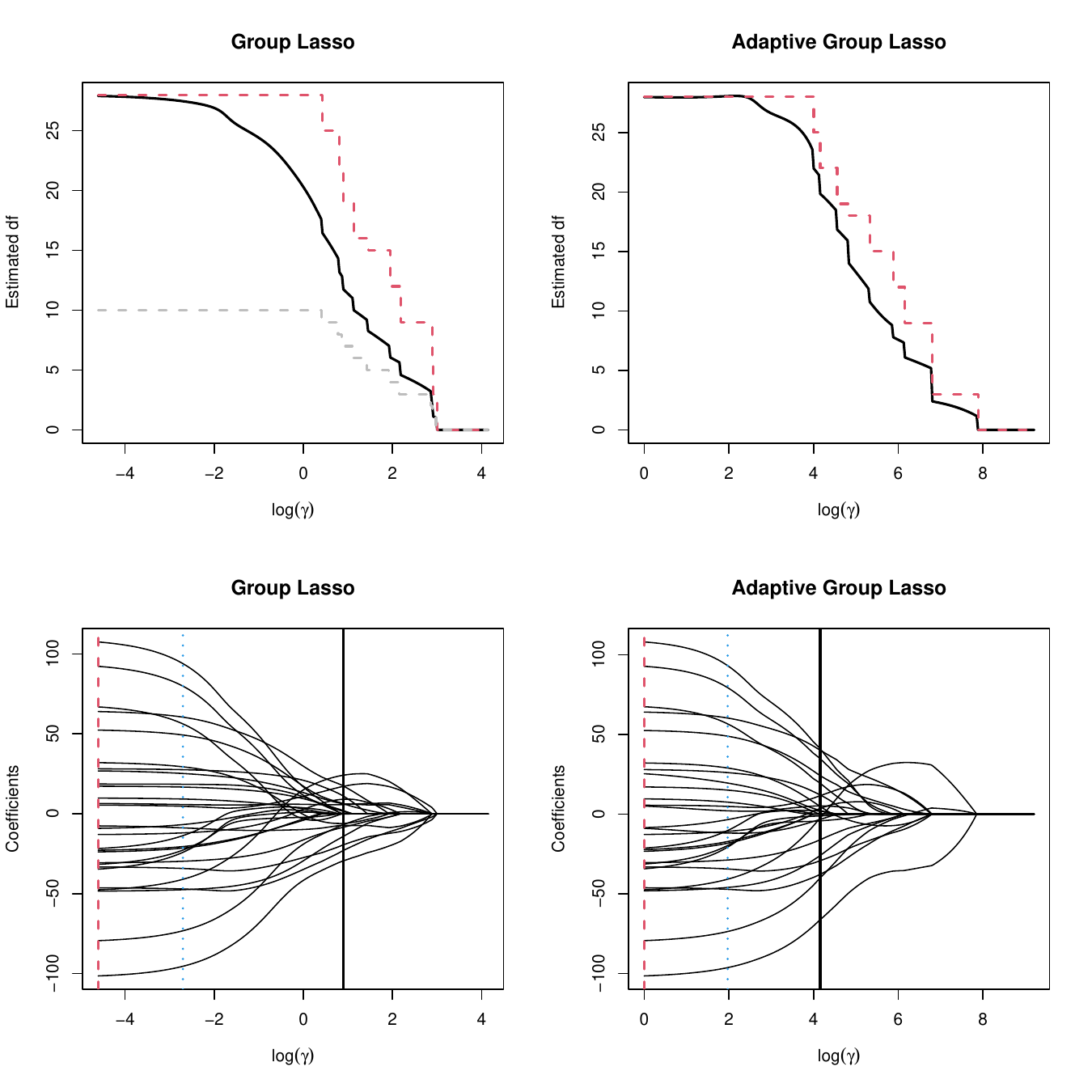}
    \caption{Results of the Group Lasso (left) and Adaptive Group  Lasso (right) estimation for the Diabetes data with the covariates discretized and encoded as dummy variables. Estimated degrees of freedom (upper panels, continuous lines) along with the size of the active set (dashed lines). For the Group Lasso also the size of the active groups is reported (dashed gray lines). Complete solution path (lower panels) with vertical lines denoting the best $\gamma$ according to correct BIC (continuous line), cross validation (dotted line), BIC with active set size as estimator of the degrees of freedom (dashed).}
    \label{fig:diabetes_small_group}
\end{figure*}
an incomplete characterization of the information utilized in parameter estimation. Notably, the upper right panel also demonstrates the behavior described by Corollary \ref{coro:linearincreasing}. Additionally, it can be noticed that, in some interval $(\gamma_l, \gamma_{l+1})$ the value of the estimated degrees of freedom in a left neighborhood of $\gamma_{l+1}$ is actually higher than that for a left neighborhood of $\gamma_{l}$ as also discussed as a comment of Corollary \ref{coro:linearincreasing}.
The lower panels of Figure \ref{fig:diabetes_small}, instead,  report complete solution paths for different values of $\gamma$.  Solid vertical lines denote the $\gamma$ values that minimize BIC according to Equation \eqref{eq:BIC} when using the correct estimated degrees of freedom for each criterion. For comparison, dash-dotted vertical lines mark the optimal $\gamma$ selected via leave-one-out cross-validation while, for the Adaptive Lasso, a dashed vertical line denotes the $\gamma$ value that minimizes BIC using the size of the active set in place of the correct estimated degrees of freedom. A key difference emerges in this case, as the selected  $\gamma$ is larger. This reaffirms that misrepresenting degrees of freedom via the active set size here introduces bias, likely favoring over-regularized models.
We continue our illustration discussing the results for the Group Lasso and Adaptive Group Lasso. The upper panels of Figure~\ref{fig:diabetes_small_group}, similarly to  Figure~\ref{fig:diabetes_small},  report the estimated degrees of freedom as a function of $\log(\gamma)$. Also here the  dashed lines in the up panels represent the active set. A gray dashed lines has been also added for the left up panel representing the size of the active groups.  For the Group Lasso, the correct estimated degrees of freedom (solid curve) are strictly upper-bounded by the active set size. This, contrarily to the Adaptive Lasso, reveals that using the active set size as a proxy for degrees of freedom systematically overestimate the true complexity of the model. {Moreover, as prescribed by Corollary \ref{coro:bounds_df_glasso}, estimated degrees of freedom of Group Lasso are greater than the number of active groups, presented in Figure~\ref{fig:diabetes_small_group} as a dotted gray line.} As discussed after the statement of Theorem \ref{th:df_adaglasso_nonortho}, the Adaptive Group Lasso  incorporates two competing effects: an inflation effect caused by the use of adaptive weights and a contraction effect induced by the $\ell_2$ norm. In this specific case, the result is still upper bounded by the active set even though as can be noticed by the upper right panel of Figure~\ref{fig:diabetes_small_group}. It should be noted, however, that in  other cases, the two effects almost compensate leading a to an estimation for the degrees of freedom that is not so far from the size of the active set. The lower panels of Figure~\ref{fig:diabetes_small_group} instead, tell a similar story to those of Figure~\ref{fig:diabetes_small} with the $\gamma$ chosen minimizing the correct BIC being closer to the leave-one-out cross validation choice than the one exploiting the BIC using the size of the active set in place of the correct degrees of freedom. 

\section{Conclusions}
\label{sec:conclusions}

We introduced a general framework for computing unbiased estimates of the degrees of freedom in penalized regression models, extending the foundational work of \cite{zou2007}. Our theoretical results underline that the common practice of using the size of active set as an estimate of the degrees of freedom is biased in many respect. This practice can severely distort inference since the size of active set and the true degrees of freedom are usually different, as highlighted by \cite{janson_etal.2015}. 
For the Adaptive Lasso, for example, we demonstrated that the correct estimate includes, in addition to the active set size, an adjustment term whose sign depends on the weights' selection. Under default weighting schemes, this term is positive, leading to an inflation of the degrees of freedom relative to the active set size.
In contrast, the Group Lasso exhibits a deflation effect due to its reliance on the $\ell_2$ norm, reducing the degrees of freedom compared to the active set size. The Adaptive Group Lasso presents a more complex interplay: both inflation (from adaptive weights) and contraction (from $\ell_2$ shrinkage) coexist, making general characterization difficult. In addition to the analytical expressions for the degrees of freedom,  we provide deep understanding of their local behavior as functions of $\gamma$. Specifically, we show that for the Adaptive Lasso (under certain conditions), these exhibit strictly positive but monotonic decreasing slopes within adjacent change points. For the Group Lasso, we show that the associated degrees of freedom are bounded between the active groups and active coefficients, and their local monotonicity is intimately related to the notion of orthogonality. 
These findings advance the understanding of model complexity in penalized regression, with implications for model selection, risk estimation, and theoretical analysis of high-dimensional methods. \par
The more challenging regime $n<p$ calls for a separate theoretical investigation. In this setting, the least-squares estimator is not available, and therefore a broadly accepted definition of the adaptive weights is lacking. To the best of our knowledge, indeed, the literature does not provide a default construction of the weights of the Adaptive Lasso or Adaptive Group Lasso in this regime, precisely because their definition relies on preliminary least-squares estimates. Notably, in the Supplementary Materials we provide an empirical exploration of this high-dimensional setting. Specifically, we employ ridge regression estimates with minimal penalization as a surrogate preliminary estimator to construct adaptive weights. Although this approach is not theoretically supported, our numerical results indicate that the proposed degrees of freedom estimators remain dramatically more accurate than the common practice of approximating the degrees of freedom by the active set size.

\bigskip 

\section*{Acknowledgements}
The authors acknowledge support from the European Union- Next Generation EU, Mission 4 Component 2 via the MUR-PRIN grants- CUP C53D23002580006 ID 2022SMNNKY and CUP E53D23010290001 ID 2022KBTEBN. Mauro Bernardi also acknowledges partial funding by the BERN BIRD2222 01- BIRD 2022 grant from the University of Padua.

\begin{appendix}
\section*{Appendix: Proofs}
\label{sec:proofs}
%
\begin{proof}[Proof of Theorem \ref{th:df_adalasso_general_weights}]
Let $\bW=\diag(w_1,\dots,w_p)\in\mathbb{S}^p_{++}$. The first order conditions for the problem in Equation \eqref{eq:convex_regularized_problem_adalasso} are represented by
\begin{equation}
\label{eq:adalasso_foc}
-\bX_{\Aset}^{{\T}}(\by-\bX_{\Aset} \widehat{\bbeta}_{\Aset})+\gamma\bW_{\Aset}  \mathrm{sgn}(\widehat{\bbeta}_{\Aset})=\boldsymbol{0},
\end{equation}
where $\bX_{\Aset}$, $\widehat{\bbeta}_{\Aset}$, $\bW_{\Aset}$ and $\sgn{(\widehat{\bbeta}_{\Aset})}$ are restricted to the active set $\Aset$. It is important to note that Equation \eqref{eq:adalasso_foc} is only valid for 
$\gamma \in (\gamma_l, \gamma_{l+1})$, where $\gamma_l$ and $\gamma_{l+1}$ are two consecutive \textit{transition points} and for any $l= 1, \dots,  L-1$.
By manipulating Equation \eqref{eq:adalasso_foc} we can derive an implicit equation for $\widehat{\bbeta}_{\Aset}$,
\begin{equation*}
\widehat{\bbeta}_{\Aset} = (\bX_{\Aset}^{\T} \bX_{\Aset})^{-1}(\bX_{\Aset}^{\T} \by - \gamma \bW_{\Aset} \text{sgn}(\widehat{\bbeta}_{\Aset})),
\end{equation*}
and denoting
\begin{equation*}
\label{eq:dof_adalasso_definitions}
\bH_\gamma(\by)=\bX_{\Aset}(\bX_{\Aset}^{\T}\bX_{\Aset})^{-1} \bX_{\Aset}^{\T}, \quad \quad 
\eta_\gamma(\by)=\bX_{\Aset}(\bX_{\Aset}^{\T}\bX_{\Aset})^{-1}\bW_{\Aset}\sgn(\widehat{\bbeta}_{\Aset}),
\end{equation*}
we can also derive the equation linking $\widehat{\by}_\gamma$ and $\by$, i.e.
\begin{equation*}
\widehat{\by}_\gamma=\bX\betahat=\bH_\gamma(\by)\by-\gamma\eta_\gamma(\by).
\end{equation*}
In order to apply Equation \eqref{eq:df_stein}, we need $\partial \widehat{\by}_\gamma / \partial \by$. Consider the increments $\by+\Delta \by$,
then
\begin{equation*}
\widehat{\by}_\gamma(\by+\Delta\by)=\bX \widehat{\bbeta}(\by+\Delta \by)=\bH_\gamma(\by+\Delta \widehat{\by}_\gamma)(\by+\Delta \by)-\gamma \eta_\gamma(\by+\Delta\by). 
\end{equation*}
If the increment of $\by$ is small enough, e.g. $\vert\Delta \by \vert <\varepsilon$, \cite{zou2007} showed that for Lasso the projection matrix $\bH_\gamma$ and the function $\eta_\gamma$ remain constant. Indeed, if $\gamma$ is not a transition point, small perturbations of $\by$ does not affect neither $\bH_\gamma$ nor $\eta_\gamma$. In our case, instead, we have that $\bH_\gamma(\by+\Delta \by)=\bH_\gamma(\by)$ but $\eta_\gamma(\by+\Delta \by) \neq\eta_\gamma(\by)$, because $\bW_\Aset(\by+\Delta \by) \neq\bW_\Aset(\by)$.
Therefore
\begin{equation*}
\begin{aligned}
\widehat{\by}_\gamma(\by+\Delta\by) & =\bH_\gamma(\by)(\by+\Delta\by)-\gamma \eta_\gamma(\by+\Delta\by)\\
\frac{\widehat{\by}_\gamma(\by+\Delta\by)-\widehat{\by}_\gamma(\by)}{\Delta \by} & =\frac{\bH_\gamma(\by) \Delta \by-\gamma\big(\eta_\gamma(\by+\Delta\by)-\eta(\by)\big)}{\Delta \by},
\end{aligned}
\end{equation*}
and
\begin{equation}
\label{eq:H_minus_eta}
\frac{\partial \widehat{\by}_\gamma}{\partial \by}=\bH_\gamma(\by)-\gamma \frac{\partial\eta_\gamma(\by)}{\partial \by}.
\end{equation}
By applying the trace operator, we have
\begin{equation*}
\begin{aligned}
\widehat{df}_\gamma &= \trace\bigg(\frac{\partial \widehat{\by}_\gamma}{\partial \by}\bigg) = \trace \bigg( \bH_\gamma(\by)-\gamma \frac{\partial\eta_\gamma(\by)}{\partial \by} \bigg)\\
&= \trace(\bX_{\Aset}(\bX_{\Aset}^{\T}\bX_{\Aset})^{-1} \bX_{\Aset}^{\T}) - \gamma \trace \bigg( \frac{\partial\bP_{\Aset}\bW_{\Aset}\sgn(\widehat{\bbeta}_{\Aset})}{\partial \by} \bigg),
\end{aligned}
\end{equation*}
where $\bP=\bX(\bX^{\T}\bX)^{-1}\in\mathbb{R}^{n\times p}$ and $\bP_{\Aset}=\bX_{\Aset}(\bX_{\Aset}^{\T}\bX_{\Aset})^{-1}\in\mathbb{R}^{n\times \vert\Aset\vert}$. Then, direct application of Lemma \ref{lemma:adalasso_weights_derivative} in the Supplementary Materials, yields to:
\begin{equation*}
\begin{aligned}
\widehat{df}_\gamma &= |\Aset|-\gamma \trace \bigg( \sum_{j \in \Aset} \sgn(\widehat{\beta}_j) \sgn(\widehat{\beta}_j^\mathsf{LS}) \frac{\partial w_j(z)}{\partial z} \bigg|_{z = \widehat{\beta}_j^\mathsf{LS}} \bp_{j} \bp_{\Aset,\pi(j)}^{\T} \bigg) \\
&= |\Aset|-\gamma \sum_{j \in \Aset} \sgn(\widehat{\beta}_j) \sgn(\widehat{\beta}_j^\mathsf{LS}) \frac{\partial w_j(z)}{\partial z} \bigg|_{z = \widehat{\beta}_j^\mathsf{LS}} 
\trace ( \bp_{j} \bp_{\Aset,\pi(j)}^{\T})\\
&= |\Aset|-\gamma \sum_{j \in \Aset} \sgn(\widehat{\beta}_j) \sgn(\widehat{\beta}_j^\mathsf{LS}) \frac{\partial w_j(z)}{\partial z} \bigg|_{z = \widehat{\beta}_j^\mathsf{LS}} 
\trace( \bp_{\Aset,\pi(j)}^{\T} \bp_{j}),
\end{aligned}
\end{equation*}
where $\bp_{j}$ and $\bp_{\Aset,\pi(j)}$ denote the $j$-th column of the matrix $\bP$ and the $\pi(j)$-th column of the matrix $\bP_\Aset$, respectively. Note that we can write $\bp_j = \bP \be_{j}$ and $\bp_{\Aset,\pi(j)} = \bP_{\Aset} \be_{\pi(j)}$, where $\be_j$ denotes a column vector with one in position $j$ and zeros elsewhere. Let $\bS_{\Aset}\in\{0,1\}^{\vert\Aset\vert\times p}$ denote the selection matrix obtained from the identity matrix $\bI_p$ by retaining only the rows corresponding to the index set $\Aset$, i.e. $\bS_{\Aset}\equiv\bI_{[\Aset,]}$. Then $\bX_{\Aset}=\bX \bS_{\Aset}^{\T}$ and
$\bX_\Aset^{\T}\bP\be_{j}=\bX_\Aset^{\T}\bX(\bX^{\T}\bX)^{-1}\be_{j}=\bS_\Aset\bX^{\T}\bX(\bX^{\T}\bX)^{-1}\be_{j}=\bS_\Aset\be_{j}=\be_{\pi(j)}$. Therefore 
\begin{equation*}
\bp_{\Aset, \pi(j)}^{\T} \bp_{j} = \be_{\pi(j)}^{\T} (\bX_{\Aset}^{\T}\bX_{\Aset})^{-1}\bX_{\Aset}^{\T} \bX(\bX^{\T}\bX)^{-1} \be_{j} = \be_{ \pi(j)}^{\T} (\bX_{\Aset}^{\T}\bX_{\Aset})^{-1} \be_{\pi(j)} = [(\bX_{\Aset}^{\T}\bX_{\Aset})^{-1}]_{\pi(j),\pi(j)},
\end{equation*}
and
\begin{equation*}
\widehat{df}_\gamma =|\Aset|-\gamma \sum_{j \in \Aset} \sgn(\widehat{\beta}_j) \sgn(\widehat{\beta}_j^\mathsf{LS}) \frac{\partial w_j(z)}{\partial z} \bigg|_{z = \widehat{\beta}_j^\mathsf{LS}} [(\bX_{\Aset}^{\T} \bX_{\Aset} )^{-1}]_{\pi(j),\pi(j)},
\end{equation*}
to arrive to the statement of the theorem. In the orthogonal case $\sgn(\widehat{\beta}_j) \sgn(\widehat{\beta}_j^\mathsf{LS})=1$ and $(\bX_{\Aset}^{\T} \bX_{\Aset} )^{-1}_{\pi(j),\pi(j)}$ $=1$ for each $j$, which simplifies the above result to
\begin{equation*}
\widehat{df}_\gamma =  |\Aset|-\gamma \sum_{j \in \Aset}  \frac{\partial w_j(z)}{\partial z} \bigg|_{z = \widehat{\beta}_j^\mathsf{LS}} .
\end{equation*}
\end{proof}
%
\begin{proof}[Proof of Corollary~\ref{th:df_adalasso_ortho}]
Under the setting of Theorem \ref{th:df_adalasso_general_weights} and the orthonormal design assumption, we have that $\bH_\gamma(\by) = \bX_{\Aset} \bX_{\Aset}^{\T}$, $\eta_\gamma(\by) = \bX_{\Aset}\bW_{\Aset}(\by)\sgn(\betahat_{\Aset})$ and $\widehat{\bbeta}^{\mathsf{LS}} = \bX^{\T} \by$. Moreover, if the weights are chosen as the inverse of the absolute values of least squares estimates the matrix $\bW$ is given by
\begin{equation*}
\begin{aligned}
\bW=\mathrm{diag}\Big(\frac{1}{\vert\betahat^{\mathsf{LS}}\vert}\Big)=\mathrm{diag}\Big(\frac{1}{\vert\bX^{\T}\by\vert}\Big),
\end{aligned}
\end{equation*}
with generic term being $w_{jj} = w_j(|\widehat{\beta}_j^\mathsf{LS}|) =  1/|\widehat{\beta}^{\mathsf{LS}}_j| =1/|\bx_j^{\T} \by|$. By applying the result of Theorem \ref{th:df_adalasso_general_weights}
in the orthonormal case, we have 
\begin{equation*}
\frac{\partial w_j(z)}{\partial z} \bigg|_{z=\widehat{\beta}_j^\mathsf{LS}} = \frac{\partial}{\partial z} \bigg( \frac{1}{z}\bigg)\bigg|_{z=\widehat{\beta}_j^\mathsf{LS}} =  - \frac{1}{z^2}\bigg|_{z=\widehat{\beta}_j^\mathsf{LS}} =  - \frac{1}{(\widehat{\beta}_j^\mathsf{LS})^2},
\end{equation*}
which concludes the proof.
\end{proof}
%
\begin{proof}[Proof of Corollary \ref{th:df_adalasso_nonortho}]
Under the setting of Theorem \ref{th:df_adalasso_general_weights}, we have $\widehat{\bbeta}^{\mathsf{LS}} = (\bX^{\T}\bX)^{-1}\bX^{\T} \by$ and
\begin{equation*}
\begin{aligned}
\bW=\mathrm{diag}\Big(\frac{1}{\vert\betahat^{\mathsf{LS}}\vert}\Big)=\mathrm{diag}\Big(\frac{1}{\vert(\bX^{\T}\bX)^{-1}\bX^{\T}\by\vert}\Big),
\end{aligned}
\end{equation*}
with generic term being $w_{jj} = w_j(|\widehat{\beta}_j^\mathsf{LS}|) = 1/|\widehat{\beta}^{\mathsf{LS}}_j| =1/|\bp_j^{\T} \by|$, where $\bp_j\in\mathbb{R}^n$ is the $j$-th column of $\bP=\bX(\bX^{\T}\bX)^{-1}\in\mathbb{R}^{n\times p}$. By applying the result of Theorem \ref{th:df_adalasso_general_weights} in the non-orthonormal case, we have
\begin{equation*}
\frac{\partial w_j(z)}{\partial z} \bigg|_{z=\widehat{\beta}_j^\mathsf{LS}} = \frac{\partial}{\partial z} \bigg( \frac{1}{z}\bigg)\bigg|_{z=\widehat{\beta}_j^\mathsf{LS}} =  - \frac{1}{z^2}\bigg|_{z=\widehat{\beta}_j^\mathsf{LS}} =  - \frac{1}{(\widehat{\beta}_j^\mathsf{LS})^2},
\end{equation*}
which concludes the proof.
\end{proof}
%
\begin{proof}[Proof of corollary \ref{coro:linearincreasing}]
Define a mapping $\pi: \Aset \rightarrow\{1,2, \ldots,|\Aset|\}$ such that for each $j \in \Aset$, $\pi(j)=i$ if $\Aset_i=j$. We first prove that $b_l>0$ for each $l$ and then that $b_l > b_{l+1}$ for each $l$.
To prove the first inequality, note that the estimated degrees of freedom are given by a linear function $|\Aset| + b_l\gamma$, where the slope $b_l$ is, by Theorem \ref{th:df_adalasso_general_weights}, equal to:
\begin{equation*} 
b_l = -\sum_{j \in \Aset} \sgn({\widehat{\beta}_j})\sgn({\widehat{\beta}_j^\mathsf{LS}})\frac{\partial w(z)}{\partial z}\bigg|_{x=\widehat{\beta}_j^\mathsf{LS}}[(\bX_{\Aset}^{\T} \bX_{\Aset} )^{-1}]_{\pi(j),\pi(j)}.
\end{equation*} 
This quantity is strictly positive because by theorem's assumptions we must have $\Aset \neq \emptyset$, signs concordance and $w^\prime(|\beta_j^\mathsf{LS}|)<0$, and $[(\bX_{\Aset}^{\T} \bX_{\Aset} )^{-1}]_{\pi(j),\pi(j)}>0$.\par
To prove the second inequality, define with $\Aset_l$ the active set for $\gamma \in (\gamma_l, \gamma_{l+1})$ and with $\Aset_{l+1}$ the active set for $\gamma \in (\gamma_{l+1}, \gamma_{l+2})$. Under the assumptions of the Corollary, we have $\Aset_{l+1} \subset \Aset_{l}$. Assume without loss of generality that moving from $\Aset_{l}$ to $\Aset_{l+1}$ the $k$-th coefficient leaves the active set i.e., $\Aset_{l} \setminus \Aset_{l+1} = \beta_k$. Then,
\begin{align*}
b_{l} &= -\sum_{j \in \Aset_l} \sgn({\widehat{\beta}_j})\sgn({\widehat{\beta}_j^\mathsf{LS}})\frac{\partial w(z)}{\partial z}\bigg|_{x=\widehat{\beta}_j^\mathsf{LS}}[(\bX_{\Aset}^{\T} \bX_{\Aset} )^{-1}]_{\pi(j),\pi(j)}\\
&=-\sum_{j \in \Aset_{l+1}} \sgn({\widehat{\beta}_j})\sgn({\widehat{\beta}_j^\mathsf{LS}})\frac{\partial w(z)}{\partial z}\bigg|_{x=\widehat{\beta}_j^\mathsf{LS}}[(\bX_{\Aset}^{\T} \bX_{\Aset} )^{-1}]_{\pi(j),\pi(j)}\nonumber\\
&\qquad\qquad-\sgn({\widehat{\beta}_k})\sgn({\widehat{\beta}_k^\mathsf{LS}})\frac{\partial w(z)}{\partial z}\bigg|_{x=\widehat{\beta}_k^\mathsf{LS}}[(\bX_{\Aset}^{\T} \bX_{\Aset} )^{-1}]_{\pi(k),\pi(k)}= b_{l+1} + c, 
\end{align*}
with $c>0$, concluding the proof.
\end{proof}
%
\begin{proof}[Proof of Theorem \ref{th:df_glasso_ortho}]
Let $\bW=\diag(w_1,\dots,w_G)\in\mathbb{S}^G_{++}$. The first order conditions for the problem in Equation \eqref{eq:convex_regularized_problem_glasso} are represented by
\begin{equation*}
-\bX_{\Aset_G}^{\T}(\by-\bX_{\Aset_G}\betahat_{\Aset_G})+\gamma\widetilde{\bbeta}_{\Aset_G}=\boldsymbol{0},
\end{equation*}
where $\bX_{\Aset_G}$ and $\widehat{\bbeta}_{\Aset_G}$ are restricted to the active set $\Aset_G$, and $\widetilde{\bbeta}_{\Aset_G}$ is a column vector whose entries are $ \Big[w_g\widehat{\bbeta}_g / \lVert \widehat{\bbeta}_g \rVert_2\Big]_{g\in\Aset_G}$. Similarly to Lasso and Adaptive Lasso, by manipulating terms recalling the orthonormal design assumption, we arrive to the following expression:
\begin{equation*}
\betahat_{\Aset_G} = (\bX_{\Aset_G}^{{\T}} \bX_{\Aset_G})^{-1} \big(\bX_{\Aset_G}^{\T} \by - \gamma\widetilde{\bbeta}_{\Aset_G} \big) = \bX_{\Aset_G}^{\T} \by - \gamma\widetilde{\bbeta}_{\Aset_G},
\end{equation*}
and denoting
\begin{equation*}
\label{eq:dof_glasso_definitions}
\bH_\gamma(\by)=\bX_{\Aset_G}\bX_{\Aset_G}^{\T},\qquad
\eta_\gamma(\by)=\bX_{\Aset_G}\widetilde{\bbeta}_{\Aset_G},
\end{equation*}
we get the expression linking $\widehat{\by}_\gamma$ and $\by$
\begin{equation*}
\widehat{\by}_\gamma=\bX\betahat=\bH_\gamma(\by)\by-\gamma\eta_\gamma(\by).
\end{equation*}
Considering the increments $\by+\Delta \by$, we have again $\eta_\gamma(\by+\Delta \by) \neq\eta_\gamma(\by)$ and arrive again at Equation \eqref{eq:H_minus_eta}. However, for Group Lasso is not easy to work with this expression directly, as $\eta_\gamma$ is function of $\by$ in an indirect way, since the dependence on $\by$ only appears through $\widehat{\bbeta}$. Despite that, exploiting the chain rule, we insert $\partial \widehat{\by}_\gamma / \partial \by$ in the right side of the equation allowing us to write
\begin{equation*}
\frac{\partial \widehat{\by}_\gamma}{\partial \by}=\bH_\gamma(\by)-\gamma \frac{\partial\eta_\gamma(\by)}{\partial \widehat{\by}}\frac{\partial \widehat{\by}_\gamma}{\partial\by},
\end{equation*}
and, by isolating the quantity of interest, obtaining
\begin{equation}
\label{eq:I_plus_eta_inverse_H}
\frac{\partial \widehat{\by}_\gamma}{\partial\by}=\left(\bI_n+\gamma \frac{\partial\eta_\gamma(\by)}{\partial \widehat{\by}}\right)^{-1} \bH_\gamma(\by). 
\end{equation}
In order to compute $\partial \eta_\gamma(\by)/\partial \widehat{\by}$ we make use again of the chain rule:
\begin{equation*}
\begin{aligned}
\frac{\partial\eta_\gamma(\by)}{\partial \widehat{\by}}&=  \frac{\partial\eta_\gamma(\by)}{\partial \betahat_{\Aset_G}}\frac{\partial\betahat_{\Aset_G}}{\partial \widehat{\by}}.
\end{aligned}
\end{equation*}
By applying Lemma \ref{lemma:product_Ax_norm_fd} in the Supplementary Materials, to the former term, we obtain:
\begin{equation*}
\begin{aligned}
\frac{\partial \eta_\gamma(\by)}{\partial\betahat_{\Aset_G}} &= \begin{bmatrix}\frac{\partial \eta_\gamma(\by)}{\partial\betahat_1}&  \cdots&\frac{\partial \eta_\gamma(\by)}{\partial\betahat_G}  \end{bmatrix}= \begin{bmatrix}\frac{\partial}{\partial \betahat_1} \bigg(\frac{ w_1 \bX_1\betahat_1}{\| \betahat_1 \|_2}\bigg) &\cdots& \frac{\partial}{\partial \betahat_G} \bigg(\frac{ w_G \bX_G\betahat_G}{\| \betahat_G \|_2}\bigg)
\end{bmatrix}\\
&= \begin{bmatrix} w_1 \bX_1\frac{\partial}{\partial\betahat_1}\Big(\frac{ \betahat_1}{\| \betahat_1 \|_2}\Big) &\cdots& w_G \bX_G\frac{\partial}{\partial\betahat_G}\Big(\frac{ \betahat_G}{\| \betahat_G \|_2}\Big)
\end{bmatrix}\\
&= \begin{bmatrix}w_1\bX_1\Bigg(\frac{\bI_{p}}{\Vert\betahat_1\Vert_2}-\frac{\betahat_1\betahat_1^{\T}}{\Vert\betahat_1\Vert_2^3}\Bigg) & \cdots & w_G\bX_G\Bigg(\frac{\bI_{p}}{\Vert\betahat_G\Vert_2}-\frac{\betahat_G\betahat_G^{\T}}{\Vert\betahat_G\Vert_2^3}\Bigg)
\end{bmatrix},
\end{aligned}
\end{equation*}
and, by defining
\begin{equation}
\label{eq:Pi_Aset_definition}
\bPi_g = w_g \Bigg(\frac{\bI_{p}}{\Vert\betahat_g\Vert_2}-\frac{\betahat_g\betahat_g^{\T}}{\Vert\betahat_g\Vert_2^3}\Bigg) \quad \text{and} \quad \bPi_{\Aset_G} = \mathrm{blockdiag}(\bPi_1, \ldots, \bPi_G),
\end{equation}
we get
\begin{equation*}
\begin{aligned}
\frac{\partial \eta_\gamma(\by)}{\partial\betahat_{\Aset_G}} 
&=\begin{bmatrix}
\bX_1\bPi_1 & \cdots &\bX_g\bPi_g&\cdots&\bX_G\bPi_G
\end{bmatrix}\\
&= \begin{bmatrix}
\bX_1 & \cdots &\bX_g&\cdots&\bX_G
\end{bmatrix}
\begin{bmatrix}
\bPi_1 & \cdots & 0 & \cdots & 0\\
0 & \cdots & \bPi_g & \cdots & 0 \\
0 & \cdots & 0 & \cdots &\bPi_G
\end{bmatrix}= \bX_{\Aset_G} \bPi_{\Aset_G}.
\end{aligned}
\end{equation*}
For the latter, since $\widehat{\by}=\bX_{\Aset_G}\betahat_{\Aset_G}$, we have
\begin{equation*}
\begin{aligned}
\frac{\partial\betahat_{\Aset_G}}{\partial \widehat{\by}}=  
\bigg(\frac{\partial \widehat{\by}}{\partial\betahat_{\Aset_G}}\bigg)^{-1} = \bigg( \frac{\partial \bX_{\Aset_G}\widehat{\bbeta}_{\Aset_G}}{\partial \widehat{\bbeta}_{\Aset_G}}\bigg)^{-1} = \bX_{\Aset_G}^-=\bX_{\Aset_G}^{\T}.
\end{aligned}
\end{equation*}
The final expression for $\partial\eta_\gamma(\by) / \partial \widehat{\by}$ is therefore
\begin{equation*}
\frac{\partial\eta_\gamma(\by)}{\partial \widehat{\by}}= 
\frac{\partial \eta_\gamma(\by)}{\partial\betahat_{\Aset_G}} \frac{\partial\betahat_{\Aset_G}}{\partial \widehat{\by}}= \bX_{\Aset_G} \bPi_{\Aset_G} \bX_{\Aset_G}^{\T}.
\end{equation*}
To compute degrees of freedom, from Equation \eqref{eq:I_plus_eta_inverse_H}, we write
\begin{equation*}
\begin{aligned}
\widehat{df}_\gamma &= \trace \bigg(\frac{\partial \widehat{\by}_\gamma}{\partial\by} \bigg) = \trace \bigg\{\left(\bI_n+\gamma \frac{\partial\eta_\gamma(\by)}{\partial \widehat{\by}}\right)^{-1} \bH_\gamma(\by) \bigg\}\\
&= \trace \big\{ \big( \bI_n + \gamma \bX_{\Aset_G} \bPi_{\Aset_G} \bX_{\Aset_G}^{\T} \big)^{-1} \bX_{\Aset_G} \bX_{\Aset_G}^{\T}  \big\},
\end{aligned}
\end{equation*}
completing the proof.
\end{proof}
%
\begin{proof}[Proof of Theorem \ref{th:df_glasso_nonortho}]
The proof follows a structure similar to that in the previous one. When the design matrix is not orthogonal, we have $\bH_\gamma(\by) = \bX_{\Aset_G}(\bX_{\Aset_G}^{\T}\bX_{\Aset_G})^{-1}\bX_{\Aset_G}^{\T}$ and  $\eta_\gamma(\by) = \bX_{\Aset_G}\left(\bX_{\Aset_G}^{\T}\bX_{\Aset_G}\right)^{-1}\widetilde{\bbeta}_{\Aset_G}$. For convenience, define $\bP = \bX_{\Aset_G}(\bX_{\Aset_G}^{\T}\bX_{\Aset_G})^{-1}$ and $\bP_g$ the generic $g$-th column of $\bP$. Consequently, we can rewrite:  
\begin{equation*}
\begin{aligned}
\frac{\partial \eta_\gamma(\by)}{\partial\betahat_{\Aset_G}}&= \begin{bmatrix}
\frac{\partial \eta_\gamma(\by)}{\partial\betahat_1} & \cdots&\frac{\partial \eta_\gamma(\by)}{\partial\betahat_G}
\end{bmatrix}= \begin{bmatrix}\frac{\partial}{\partial \betahat_1} \bigg(\frac{ w_1 \bP_1\betahat_1}{\| \betahat_1 \|_2}\bigg)& \cdots &\frac{\partial}{\partial \betahat_G} \bigg(\frac{ w_G \bP_G\betahat_G}{\| \betahat_G \|_2}\bigg) \end{bmatrix} \\
&= \begin{bmatrix}
w_1 \bP_1\frac{\partial}{\partial\betahat_1}\Big(\frac{ \betahat_1}{\| \betahat_1 \|_2}\Big) & \cdots & w_G \bP_G\frac{\partial}{\partial\betahat_G}\Big(\frac{ \betahat_G}{\| \betahat_G \|_2}\Big)
\end{bmatrix} \\
&= \begin{bmatrix}
w_1\bP_1\Bigg(\frac{\bI_{p}}{\Vert\betahat_1\Vert_2}-\frac{\betahat_1\betahat_1^{\T}}{\Vert\betahat_1\Vert_2^3}\Bigg) & \cdots & w_G\bP_G\Bigg(\frac{\bI_{p}}{\Vert\betahat_G\Vert_2}-\frac{\betahat_G\betahat_G^{\T}}{\Vert\betahat_G\Vert_2^3}\Bigg)
\end{bmatrix},
\end{aligned}
\end{equation*}
from which we get
\begin{equation*}
\begin{aligned}
\frac{\partial \eta_\gamma(\by)}{\partial\betahat_{\Aset_G}} &= \begin{bmatrix}
\bP_1 \bPi_1 & \cdots & \bP_g \bPi_g & \cdots & \bP_G \bPi_G
\end{bmatrix}\\
&= \begin{bmatrix}
\bP_1  & \cdots & \bP_g  & \cdots & \bP_G 
\end{bmatrix} \begin{bmatrix}
\bPi_1 & \cdots & 0 & \cdots & 0\\
0 & \cdots & \bPi_g & \cdots & 0 \\
0 & \cdots & 0 & \cdots &\bPi_G
\end{bmatrix}\\
&= \bP \bPi_{\Aset_G}= \bX_{\Aset_G}(\bX_{\Aset_G}^{\T}\bX_{\Aset_G})^{-1} \bPi_{\Aset_G}.
\end{aligned}
\end{equation*}
For the latter, since $\widehat{\by}=\bX_{\Aset_G}\betahat_{\Aset_G}$, we have
\begin{equation*}
\begin{aligned}
\frac{\partial\betahat_{\Aset_G}}{\partial \widehat{\by}}=  
\bigg(\frac{\partial \widehat{\by}}{\partial\betahat_{\Aset_G}}\bigg)^{-1} = \bigg( \frac{\partial \bX_{\Aset_G}\widehat{\bbeta}_{\Aset_G}}{\partial \widehat{\bbeta}_{\Aset_G}}\bigg)^{-1} = \bX_{\Aset_G}^-=(\bX_{\Aset_G}^{\T}\bX_{\Aset_G})^{-1}\bX_{\Aset_G}^{\T}.
\end{aligned}
\end{equation*}
The final expression for the derivative $\partial\eta_\gamma(\by) / \partial \widehat{\by}$, is therefore
\begin{equation*}
\frac{\partial\eta_\gamma(\by)}{\partial \widehat{\by}}=\frac{\partial \eta_\gamma(\by)}{\partial\betahat_{\Aset_G}} \frac{\partial\betahat_{\Aset_G}}{\partial \widehat{\by}} = \bX_{\Aset_G}(\bX_{\Aset_G}^{\T}\bX_{\Aset_G})^{-1} \bPi_{\Aset_G} (\bX_{\Aset_G}^{\T}\bX_{\Aset_G})^{-1}\bX_{\Aset_G}^{\T}.
\end{equation*}
To compute the degrees of freedom, we use Equation \eqref{eq:I_plus_eta_inverse_H} and write:
\begin{equation*}
\begin{aligned}
\widehat{df}_\gamma &= \trace \bigg(\frac{\partial \widehat{\by}_\gamma}{\partial\by} \bigg) = \trace \bigg\{\left(\bI_n+\gamma \frac{\partial\eta_\gamma(\by)}{\partial \widehat{\by}}\right)^{-1} \bH_\gamma(\by) \bigg\}\\
&= \trace \bigg\{ \bigg( \bI_n + \gamma \bX_{\Aset_G}(\bX_{\Aset_G}^{\T} \bX_{\Aset_G})^{-1} \bPi_{\Aset_G} (\bX_{\Aset_G}^{\T} \bX_{\Aset_G})^{-1}\bX_{\Aset_G}^{\T} \bigg)^{-1} \bX_{\Aset_G} (\bX_{\Aset_G}^{\T} \bX_{\Aset_G})^{-1}\bX_{\Aset_G}^{\T}  \bigg\},
\end{aligned}
\end{equation*}
which completes the proof.
\end{proof}
%
\begin{proof}[Proof of corollary \ref{coro:bounds_df_glasso}]
We begin by noting that both the matrices $\bA$ and $\bB$ introduced in Theorems \ref{th:df_glasso_ortho} and  \ref{th:df_glasso_nonortho} are positive semidefinite. The matrix $\bA$ has its first  $|\Aset_p|$ eigenvalues equal to $1$, with all remaining eigenvalues equal to $0$. The matrix $\bB$ is also positive semidefinite, since each component matrix $\bPi_g$ is. In fact, since $w_g/\lVert \bbeta_g\rVert_2>0$, and the matrix $\bI_{n_g} - \bbeta_g \bbeta_g^{\T} / \lVert \bbeta_g\rVert_2^2$ has $n_g-1$ eigenvalues equal to $1$ and one equal to $0$, each $\bPi_g$ inherits this semi-definiteness. When $\bX$ is orthonormal, the eigenvalues of $\bB$ are  
\begin{equation*}
\lambda(\bB) = \bigg(\bigcup_{g \in \Aset} \bigg\{ \underbrace{\frac{w_g}{\lVert \widehat{\bbeta}_g \rVert_2}, \ldots, \frac{w_g}{\lVert \widehat{\bbeta}_g \rVert_2}}_{n_g-1}, 0 \bigg\}, \underbrace{0, \ldots, 0 }_{n-|\Aset_p|}\bigg),
\end{equation*}
for a total of $|\Aset_p| - |\Aset_G| $ eigenvalues greater than zero and the remaining eigenvalues equal to zero. In the non orthonormal case, the expression of the eigenvalues of $\bB$ is not known, but we still have $|\Aset_p| - |\Aset_G| $ eigenvalues greater than zero and the remaining eigenvalues equal to zero. When the design is orthonormal,
\begin{equation}
\label{eq:trace_as_ev}
\text{trace} ((\bI_n+\gamma\bB)^{-1}\bA) = \sum_{i=1}^n \frac{\lambda_i^A}{1+\gamma\lambda_i^B},
\end{equation}
where $\lambda_i^A$ and $\lambda_i^B$ are the eigenvalues of $\bA$ and $\bB$, respectively. It is straightforward to prove that Equation \eqref{eq:trace_as_ev} is a continuous function of $\gamma$. Given the eigenstructure of the two matrices, the first $|\Aset_p| - |\Aset_G|$ eigenvalues of $\bA$ are shrinked by the positive amount $1 + \gamma\lambda_i^B$, the following $|\Aset_G|$ eigenvalues of $\bA$ remain equal to 1 because the corresponding $\lambda_i^B$ is zero, and the remaining eigenvalues are zero. We can thus conclude that 
$$ |\Aset_G| \leq \trace ((\bI_n+\gamma\bB)^{-1}\bA) \leq \trace(\bA) = |\Aset_p|.$$
For the non orthonormal design, we apply Von Neumann's trace inequality:
\begin{equation}
\label{eq:trace_as_ev2}
|\Aset_G| \leq \sum_{i=1}^n \frac{\lambda_i^A} {1+\gamma\lambda_{i}^B} \leq \text{trace} ((\bI_n+\gamma\bB)^{-1}\bA) \leq \sum_{i=1}^n \frac{\lambda_i^A} {1+\gamma\lambda_{n-i+1}^B} \leq \sum_{i=1}^n \lambda_i^A = |\Aset_p|,
\end{equation}
thus the degrees of freedom of Group Lasso are always lower than the number of active \textit{variables} $|\Aset_p|$ and greater than the number of active \textit{groups} $|\Aset_G|$.
The equality in the previous expression is achieved if $\gamma=0$ (in which case $\widehat{df}_\gamma = p$), if $|\Aset_p|= |\Aset_G|=0$ (in which case $\widehat{df}_\gamma = 0$) or if $n_g=1$ (in which case Lasso is being fitted).
\end{proof}
%
\begin{proof}[Proof of Theorem \ref{th:monotonicity_df_glasso}]
Taking the derivative of $\widehat{df}_\gamma$ with respect to $\gamma>0$, we have
\begin{equation*}
\begin{aligned}
\frac{\mathrm{d}\widehat{df}_\gamma}{\mathrm{d}\gamma}&=
\frac{\mathrm{d}}{\mathrm{d}\gamma} \trace\left[\left(\bI_n+\gamma  \bB\right)^{-1}  \bA\right]\\
&=-\trace\left[\left( \bPi_{\Aset_G}+\gamma \frac{\mathrm{d}  \bPi_{\Aset_G}}{\mathrm{d} \gamma}\right)\bM_{\Aset_G}^{\T}\left( \bI_n+\gamma  \bB\right)^{-1}  \bA\left( \bI_n+\gamma  \bB\right)^{-1}\bM_{\Aset_G}\right],
\end{aligned}
\end{equation*}
where, to get the previous result, we used the fact that 
\begin{equation*}
\frac{\mathrm{d}\bB}{\mathrm{d}\gamma}=\bM_{\Aset_G}\frac{\mathrm{d}  \bPi_{\Aset_G}}{\mathrm{d} \gamma}\bM_{\Aset_G}^{\T}.
\end{equation*}
For the derivative $\mathrm{d}\widehat{df}_\gamma/d\gamma$ to be negative the trace should be positive. Note that as in Equation \eqref{eq:trace_as_ev2}, the matrices $\bI_n+\gamma  \bB$ and $\bA$ have non-negative eigenvalues, therefore it remains to prove that all the eigenvalues of 
\begin{equation*}
\bB+\gamma \frac{\mathrm{d}  \bB}{\mathrm{d} \gamma}=\bM_{\Aset_G}\left( \bPi_{\Aset_G}+\gamma \frac{\mathrm{d}  \bPi_{\Aset_G}}{\mathrm{d} \gamma}\right)\bM_{\Aset_G}^{\T},
\end{equation*}
are non-negative. Since  $\bPi_{\Aset_G}+\gamma \mathrm{d}  \bPi_{\Aset_G}/\mathrm{d} \gamma$ is positive semidefinite by assumption, we have the proof. 
\end{proof}
%
\begin{proof}[Proof of Theorem \ref{th:matrix_Pi_indefinite}]
Consider the derivative of $\bPi_{\Aset_G}$ wrt $\gamma$. Since $\bPi_{\Aset_G}$ is block-diagonal with blocks defined in Equation \eqref{eq:Pi_Aset_definition}, then the derivative can be computed blockwise:
\begin{equation*}
\frac{\mathrm{d}\bPi_{\Aset_G} }{\mathrm{d}\gamma}= \mathrm{blockdiag}\Big(\frac{\mathrm{d}\bPi_{1}}{\mathrm{d}\gamma}, \ldots, \frac{\mathrm{d}\bPi_{G}}{\mathrm{d}\gamma}\Big).
\end{equation*}
Applying the chain rule for the derivative, and defining $r_g=\|\betahat_{g}\|_2$, we have
\begin{equation}
\label{eq:Pi_derivative}
\frac{\mathrm{d}\bPi_{g}}{\mathrm{d}\gamma}
= -\frac{w_g}{r_g^2}\frac{\mathrm{d}r_g}{\mathrm{d}\gamma}\Big(\bI_{n_g}-\frac{3}{r_g^2} \widehat{\bbeta}_g \widehat{\bbeta}_g^{\T} \Big)-  \frac{w_g}{r_g^3} \Big( \frac{\mathrm{d}\betahat_g}{\mathrm{d}\gamma} \widehat{\bbeta}_g^{\T} + \widehat{\bbeta}_g \frac{\mathrm{d}\betahat_g}{\mathrm{d}\gamma}^{\T} \Big),\quad g \in \Aset_G,
\end{equation}
$\frac{\mathrm{d}\betahat_g}{\mathrm{d}\gamma}=\frac{\mathrm{d}\betahat_g}{\mathrm{d}\betahat_{\Aset_G}}\frac{\mathrm{d}\betahat_{\Aset_G}}{\mathrm{d}\gamma}=\bS_{g}\frac{\mathrm{d}\betahat_{\Aset_G}}{\mathrm{d}\gamma}$, and $\bS_{g}\in\{0,1\}^{n_g\times\vert\Aset_G\vert}$ denotes the selection matrix obtained from the identity matrix $\bI_{\vert\Aset_G\vert}$ by retaining only the rows corresponding to the $g$-th group. To compute the derivative $\frac{\mathrm{d}\widehat{\bbeta}_{\Aset_G}}{\mathrm{d}\gamma}$, we apply the implicit function theorem to the equation:
\begin{equation*}
F(\widehat{\bbeta}_{\Aset_G}, \gamma) =-\bX_{\Aset_G}^{\T}\left(\by-\bX_{\Aset_G}\betahat_{\Aset_G}\right)+\gamma\widetilde{\bbeta}_{\Aset_G} =\boldsymbol{0}.
\end{equation*}
Differentiating $F(\widehat{\bbeta}_{\Aset_G}, \gamma)$ with respect to $\gamma$ and solving for $\frac{\mathrm{d}\widehat{\bbeta}_{\Aset_G}}{\mathrm{d} \gamma}$ gives:
\begin{equation}
\label{eq:IFT_Pi}
\frac{\mathrm{d}\widehat{\bbeta}_{\Aset_G}}{\mathrm{d}\gamma} = -\left( \frac{\partial F}{\partial \widehat{\bbeta}_{\Aset_G}} \right)^{-1}  \frac{\partial F}{\partial \gamma}\in\mathbb{R}^{\vert\Aset_G\vert},
\end{equation}
where
\begin{equation*}
\frac{\partial F}{\partial \widehat{\bbeta}_{\Aset_G}} = \bX_{\Aset_G}^{\T} \bX_{\Aset_G}+ \gamma \bPi_{\Aset_G}\in\mathbb{S}^{\vert\Aset_G\vert}_+,
\quad
\frac{\partial F}{\partial \gamma} =\widetilde{\bbeta}_{\Aset_G}\in\mathbb{R}^{\vert\Aset_G\vert},
\end{equation*}
and, substituting in Equation \eqref{eq:IFT_Pi}, we get the final expression for the derivative:
\begin{equation}
\label{eq:IFT_Pi_FINAL}
\frac{\mathrm{d}\widehat{\bbeta}_{\Aset_G}}{\mathrm{d}\gamma}
= -\left( \bX_{\Aset_G}^{\T} \bX_{\Aset_G} + \gamma \bPi_{\Aset_G} \right)^{-1}\widetilde{\bbeta}_{\Aset_G}=-\bR_{\Aset_G}\widetilde{\bbeta}_{\Aset_G}=-\bR_{\Aset_G}^\star\widehat{\bbeta}_{\Aset_G}\in\mathbb{R}^{\vert\Aset_G\vert}.
\end{equation}
where $\bR_{\Aset_G}^\star=\bR_{\Aset_G}\bW_{\Aset_G},\bW_{\Aset_G}=\diag\Big\{\frac{w_g}{\Vert\widehat{\bbeta}_{g}\Vert_2}\boldsymbol{\iota}_{n_g}\Big\}_{g\in\Aset_G}$. Moreover, letting $r_{\Aset_G} = \Vert \widehat{\bbeta}_{\Aset_G}\Vert_2$ and $r_g= \Vert \widehat{\bbeta}_{g}\Vert_2$, which are functions of $\gamma$, we aim to compute their derivative with respect to $\gamma$. Using the chain rule, we obtain the following:
\begin{equation*}
\frac{\mathrm{d}r_{\Aset_G}}{\mathrm{d}\gamma}
= \frac{\mathrm{d}}{\mathrm{d}\gamma} \left( \Vert \widehat{\bbeta}_{\Aset_G} \Vert_2 \right)
= \frac{1}{\Vert \widehat{\bbeta}_{\Aset_G} \Vert_2}\widehat{\bbeta}_{\Aset_G}^{\T} \frac{\mathrm{d}\widehat{\bbeta}_{\Aset_G}}{\mathrm{d}\gamma}.
\end{equation*}
Substituting the expression for \( \frac{\mathrm{d}\widehat{\bbeta}_{\Aset_G}}{\mathrm{d}\gamma} \) from Equation~\eqref{eq:IFT_Pi_FINAL}, we get:
\begin{equation*}
\frac{\mathrm{d}r_{\Aset_G}}{\mathrm{d}\gamma}
= \frac{1}{\Vert \widehat{\bbeta}_{\Aset_G} \Vert_2}
\left( -\widehat{\bbeta}_{\Aset_G}^{\T} \bR_{\Aset_G} \widetilde{\bbeta}_{\Aset_G} \right)=-\frac{ \widehat{\bbeta}_{\Aset_G}^{\T} \bR_{\Aset_G} \widetilde{\bbeta}_{\Aset_G} }{ \Vert \widehat{\bbeta}_{\Aset_G} \Vert_2}=-\frac{ \widehat{\bbeta}_{\Aset_G}^{\T} \bR_{\Aset_G}^\star \widehat{\bbeta}_{\Aset_G} }{\Vert \widehat{\bbeta}_{\Aset_G} \Vert_2},
\end{equation*}
and
\begin{equation}
\label{eq:dr_gdgamma}
\frac{\mathrm{d}r_g}{\mathrm{d}\gamma}
= \frac{\mathrm{d}}{\mathrm{d}\gamma} \left( \Vert \widehat{\bbeta}_{g} \Vert_2 \right)
= \frac{1}{\Vert \widehat{\bbeta}_{g} \Vert_2}\widehat{\bbeta}_{g}^{\T} \frac{\mathrm{d}\widehat{\bbeta}_{g}}{\mathrm{d}\gamma}=-\frac{1}{\Vert \widehat{\bbeta}_{g} \Vert_2}\widehat{\bbeta}_{g}^{\T} \Big[\bR_{\Aset_G}^\star\betahat_{\Aset_G}\Big]_g.
\end{equation}
Substituting the expression for $\frac{\mathrm{d}\betahat_g}{\mathrm{d}\gamma}$ into Equation 
\eqref{eq:Pi_derivative}, we get:
\begin{equation}
\label{eq:Pi_derivative_final}
\frac{\mathrm{d}\bPi_{g}}{\mathrm{d}\gamma}
= -\frac{w_g}{r_g^2}\frac{\mathrm{d}r_g}{\mathrm{d}\gamma}\Big(\bI_{n_g}-\frac{3}{r_g^2} \widehat{\bbeta}_g \widehat{\bbeta}_g^{\T} \Big)+  \frac{w_g}{r_g^3} \Big(\Big[\bR_{\Aset_G}^\star\widehat{\bbeta}_{\Aset_G}\Big]_g \widehat{\bbeta}_g^{\T} + \widehat{\bbeta}_g \Big[\bR_{\Aset_G}^\star\widehat{\bbeta}_{\Aset_G}\Big]_g^{\T} \Big).
\end{equation}
We are interested in finding all the eigenvalues of the matrix $\bPi_{\Aset_G}+\gamma \mathrm{d}  \bPi_{\Aset_G}/\mathrm{d} \gamma$. Let $\bu_g = \widehat{\bbeta}_g,\quad \bv_g = \left[\bR_{\Aset_G}^\star \widehat{\bbeta}_{\Aset_G} \right]_g$, the eigenvalues of the $g$-th block of $\bPi_{\Aset_G}+\gamma \mathrm{d}  \bPi_{\Aset_G}/\mathrm{d} \gamma$ can be calculated using Lemma \ref{lemma:Pi_eigenvalues_all} in the Supplementary Materials. Specifically
\begin{align*}
\lambda_{1,g} &= \frac{1}{2}(2 a_g+b_g \Vert\bu_g\Vert_2^2+2 c_g \bv_g^{\T}\bu_g) + \frac{1}{2}\sqrt{\Delta_g}, \\
\lambda_{2,g} &= a_g, \qquad \text{with multiplicity } n_g - 2, \\
\lambda_{3,g} &= \frac{1}{2}(2 a_g+b_g \Vert\bu_g\Vert_2^2+2 c_g \bv_g^{\T}\bu_g) - \frac{1}{2}\sqrt{\Delta_g},
\end{align*}
where
\begin{equation}
\label{eq:abc_definition}
a_g=\gamma\frac{w_g}{r_g^3}\bu_g^{\T}\bv_g+\frac{w_g}{r_g},\qquad
b_g=-\Big(\gamma\frac{3w_g}{r_g^5}\bu_g^{\T}\bv_g+\frac{w_g}{r_g^3}\Big),\qquad
c_g=\gamma\frac{w_g}{r_g^3}>0,
\end{equation}
and
\begin{equation}
\label{eq:Pi_eigen_Delta}
\Delta_g=4c_g^2r_g^2\Vert\bv_g\Vert_2^2+b_g^2r_g^4+4b_gc_gr_g^2\bu_g^{\T}\bv_g=-\frac{3\gamma^2w_g^2}{r_g^6}(\bu_g^{\T}\bv_g)^2+\frac{2\gamma w_g^2}{r_g^4}\bu_g^{\T}\bv_g+\frac{4\gamma^2 w_g^2\Vert\bv_g\Vert_2^2}{r_g^4}+\frac{w_g^2}{r_g^2},
\end{equation}
for $g=1,\dots,G$. Moreover, $\Delta_g\geq0$, by Lemma \ref{lemma:Delta_sign} in the Supplementary Materials. Let us now consider the sum and product of the eigenvalues $\lambda_{1,g}$ and $\lambda_{3,g}$. We have:
\begin{equation}
\label{eq:lambda1_plus_lambda3_is_lambda2}
\begin{aligned}
\lambda_{1,g}+\lambda_{3,g}&=2 a_g+b_g \Vert\bu_g\Vert_2^2+2 c_g \bv_g^{\T}\bu_g\\
&=2\Big(\gamma\frac{w_g}{r_g^3}\bu_g^{\T}\bv_g+\frac{w_g}{r_g}\Big)-\Big(\gamma\frac{3w_g}{r_g^5}\bu_g^{\T}\bv_g+\frac{w_g}{r_g^3}\Big)r_g^2+2\gamma\frac{w_g}{r_g^3}\bu_g^{\T}\bv_g\\
&=\frac{\gamma w_g}{r_g^3}\bu_g^{\T}\bv_g+\frac{w_g}{r_g}=a_g=\lambda_{2,g},
\end{aligned}
\end{equation}
therefore $\lambda_{1,g}+\lambda_{3,g}\geq0$ if $\rho_{u,v}\geq-\frac{r_g}{\gamma\Vert\bv_g\Vert_2}$, where $\rho_{u,v}=\frac{\bu_g^{\T}\bv_g}{\Vert\bu_g\Vert_2\Vert\bv_g\Vert_2}$. As concerns the product of $\lambda_{1,g}$ and $\lambda_{3,g}$, we have:
\begin{equation*}
\begin{aligned}
\lambda_{1,g}\lambda_{3,g}&=\frac{1}{4}(2 a_g+b_g \Vert\bu_g\Vert_2^2+2 c_g \bv_g^{\T}\bu_g)^2 - \frac{1}{4}\Delta_g\\
&=a_g^2+\frac{1}{4}b_g^2r_g^4+c_g^2(\bu_g^{\T}\bv_g)^2+a_gb_gr_g^2+2a_gc_g(\bu_g^{\T}\bv_g)+b_gc_gr_g^2(\bu_g^{\T}\bv_g)\\
&\qquad-c_g^2r_g^2\Vert\bv_g\vert_2^2-\frac{1}{4}b_g^2r_g^4-b_gc_gr_g^2\bu_g^{\T}\bv_g\\
&= a_g^2+a_gb_gr_g^2-c_g^2\Vert\bu_g\Vert_2^2\Vert\bv_g\Vert_2^2(1-\rho_{u,v}^2)+2a_gc_g\bu_g^{\T}\bv_g.
\end{aligned}
\end{equation*}
Now, we substitute the values of $a_g$, $b_g$ and $c_g$ in Equation \eqref{eq:abc_definition} and we get:
\begin{equation*}
\begin{aligned}
\lambda_{1,g}\lambda_{3,g}&=\Big(\gamma\frac{w_g}{r_g^3}\bu_g^{\T}\bv_g+\frac{w_g}{r_g}\Big)^2-\Big(\gamma\frac{w_g}{r_g^3}\bu_g^{\T}\bv_g+\frac{w_g}{r_g}\Big)\Big(\gamma\frac{3w_g}{r_g^5}\bu_g^{\T}\bv_g+\frac{w_g}{r_g^3}\Big)r_g^2\\
&\qquad-\gamma^2\frac{w_g^2}{r_g^6}r_g^2\Vert\bv_g\Vert_2^2(1-\rho_{u,v}^2)+2\Big(\gamma\frac{w_g}{r_g^3}\bu_g^{\T}\bv_g+\frac{w_g}{r_g}\Big)\gamma\frac{w_g}{r_g^3}\bu_g^{\T}\bv_g\\
&=\frac{\gamma^2w_g^2}{r_g^6}\Vert\bu_g\Vert_2^2\Vert\bv_g\Vert_2^2\big(\rho_{u,v}^2-1\big).
\end{aligned}
\end{equation*}
Therefore, $\lambda_{1,g}\lambda_{3,g}<0$ and the matrix $\bPi_{\Aset_G}+\gamma\mathrm{d}\bPi_{\Aset_G}/\mathrm{d}\gamma$ is indefinite, unless $\rho_{u,v}=1$ (e.g. $\bu_g\propto\bv_g$) when it is equal to $\lambda_{1,g}\lambda_{3,g}=0$. In particular, when $\bX^{\T}\bX=\bI_p$ then $\rho_{u,v}=1$ (see Corollary \ref{coro:dbeta_dgamma_orto} in Supplementary Materials) and $\lambda_{3,g}=\frac{1}{2}(2a_g+b_g r_g^2+2cv\bu_g^{\T}\bv_g)-\frac{1}{2} \sqrt{\Delta_g}$ where $\Delta_g$ is defined in Equation \eqref{eq:Pi_eigen_Delta} as:
\begin{equation*}
\begin{aligned}
\Delta_{g,\textrm{ortho}}&=-\frac{3\gamma^2w_g^2}{r_g^6}\rho_{u,v}^2r_g^2\Vert\bv_g\Vert_2^2+\frac{2\gamma w_g^2}{r_g^4}\rho_{u,v} r_g\Vert\bv_g\Vert_2+\frac{4\gamma^2w_g^2\Vert\bv_g\Vert_2^2}{r_g^4}+\frac{w_g^2}{r_g^2}\\
&=\frac{\gamma^2w_g^2\Vert\bv_g\Vert_2^2}{r_g^4}+\frac{2\gamma w_g^2\Vert\bv_g\Vert_2}{r_g^3}+\frac{w_g^2}{r_g^2}=\Bigg(\frac{\gamma w_g\Vert\bv_g\Vert_2}{r_g^2}+\frac{w_g}{r_g}\Bigg)^2=a_g^2.
\end{aligned}
\end{equation*}
Moreover, by Equation \eqref{eq:lambda1_plus_lambda3_is_lambda2}, the first part of $\lambda_{3,g}$, $2a_g+b_g r_g^2+2c_g\bu_g^{\T}\bv_g=a_g$, therefore $\lambda_{3,g}=a-\sqrt{a^2}=0$ and $\lambda_{1,g}=\lambda_{2,g}$. Also, $\bu_g^{\T}\bv_g=\Vert\bu_g\Vert_2\Vert\bv_g\Vert_2$ and $a_g$ defined in Equation \eqref{eq:abc_definition} is strictly positive, and the matrix $\bPi_{g} + \gamma \mathrm{d}\bPi_{g}/\mathrm{d}\gamma$ is positive semidefinite. To check if the orthonormal design is the only setting leading to $\rho_{u,v}=1$, we observe that $\rho_{u,v}=1$ implies $\bR_{\Aset_G}^\star\widehat{\bbeta}_{\Aset_G} = c \widehat{\bbeta}_{\Aset_G}$, which in turn means that $\widehat{\bbeta}_{\Aset_G}$ is an eigenvector of the matrix $\bR^*_{\Aset_G} = \bR_{\Aset_G}\bW_{\Aset_G} = (\bX_{\Aset_G}^{\T} \bX_{\Aset_G} + \gamma \bPi_{\Aset_G})^{-1} \bW_{\Aset_G}$, which completes the proof.
\end{proof}
%
\begin{proof}[Proof of Theorem \ref{th:df_adaglasso_general_weights}]
The first order conditions for the problem in Equation \eqref{eq:convex_regularized_problem_adaglasso} are represented by
\begin{equation*}
-\bX_{\Aset_G}^{\T}\left(\by-\bX_{\Aset_G}\betahat_{\Aset_G}\right)+\gamma \bW_{\Aset_G}\Breve{\bbeta}_{\Aset_G}=\boldsymbol{0} 
\end{equation*}
where $\bX_{\Aset_G}, \bW_{\Aset_G}$ and $\widehat{\bbeta}_{\Aset_G}$ are restricted to the active set $\Aset_G$, and $\Breve{\bbeta}_{\Aset_G} = (\widehat{\bbeta}_1 / \lVert \widehat{\bbeta}_1 \rVert_2, 
\ldots, \widehat{\bbeta}_G / \lVert \widehat{\bbeta}_G \rVert_2)$. Similarly to Lasso and Adaptive Lasso, by manipulating terms, focusing on the orthonormal design, we arrive to
\begin{equation*}
\betahat_{\Aset_G} = (\bX_{\Aset_G}^{{\T}} \bX_{\Aset_G})^{-1} \big(\bX_{\Aset_G}^{\T} \by - \gamma\bW_{\Aset_G}\Breve{\bbeta}_{\Aset_G} \big) = \bX_{\Aset_G}^{\T} \by - \gamma\bW_{\Aset_G}\Breve{\bbeta}_{\Aset_G},
\end{equation*}
and if we denote
\begin{equation*}
\bH_\gamma(\by)=\bX_{\Aset_G} \bX_{\Aset_G}^{\T},\qquad
\eta_\gamma(\by)=\bX_{\Aset_G} \bW_{\Aset_G}\Breve{\bbeta}_{\Aset_G},
\end{equation*}
we can also derive the equation linking $\widehat{\by}_\gamma$ and $\by$
\begin{equation*}
\widehat{\by}_\gamma=\bX\betahat=\bH_\gamma(\by)\by-\gamma\eta_\gamma(\by).
\end{equation*}
The general approach is to start again from Equation \eqref{eq:H_minus_eta}, and derive an expression similar to Equation \eqref{eq:I_plus_eta_inverse_H}. However,
in Adaptive Group Lasso we see that both $\bW_{\Aset_G}$ ---like Adaptive Lasso--- and $\breve{\bbeta}_{\Aset_G}$ ---like Group Lasso--- are not constant with respect to $\by$, thus a slight different approach should be employed to find the gradient of $\eta_\gamma(\by)$. In particular we make use of the product rule for differentiation in the following way. Let $\eta_0 = \bX_{\Aset_G}(\bX_{\Aset_G}^{{\T}} \bX_{\Aset_G})^{-1} $, $\eta_1(\by) = \bW_{\Aset_G}(\by)$ and $\eta_2(\by) = \Breve{\bbeta}_{\Aset_G}(\by)$ and rewrite $\eta(\by) = \eta_0 \eta_1(\by)\eta_2(\by)$. Then,
\begin{equation*}
\frac{\partial \eta_\gamma(\by)}{\partial \by} = \eta_0 \bigg \{ \bigg[ \frac{\partial \eta_1(\by)}{\partial \by} \bigg] \eta_2(\by) + \eta_1(\by) \frac{\partial \eta_2(\by)}{\partial \by} \bigg\},
\end{equation*}
where the former multiplication represents a tensor–vector contraction on the second axis, and the latter a standard matrix multiplication. By proceeding with similar arguments to those in the proof of Theorem \ref{th:df_glasso_ortho} we have
\begin{equation*}
\begin{aligned}
\frac{\partial \widehat{\by}_\gamma}{\partial \by}&=\bH_\gamma(\by)-\gamma \frac{\partial\eta_\gamma(\by)}{\partial\by}= \bH_\gamma(\by)-\gamma \eta_0 \bigg \{ \bigg[ \frac{\partial \eta_1(\by)}{\partial \by} \bigg] \eta_2(\by) + \eta_1(\by) \frac{\partial \eta_2(\by)}{\partial \by} \bigg\},\\
&= \bH_\gamma(\by)-\gamma \eta_0  \bigg[ \frac{\partial \eta_1(\by)}{\partial \by} \bigg] \eta_2(\by) -\gamma \eta_0 \eta_1(\by) \frac{\partial \eta_2(\by)}{\partial \by},\\
&= \bH_\gamma(\by)-\gamma \eta_0  \bigg[ \frac{\partial \eta_1(\by)}{\partial \by} \bigg] \eta_2(\by) -\gamma \eta_0 \eta_1(\by) \frac{\partial \eta_2(\by)}{\partial \widehat{\by}}\frac{\partial \widehat{\by}_\gamma}{\partial\by},
\end{aligned}
\end{equation*}
and, by isolating the quantity of interest, we obtain
\begin{equation}
\label{eq:I_plus_eta2_inverse_H_minus_eta1}
\frac{\partial \widehat{\by}_\gamma}{\partial\by}=\left(\bI_n+\gamma \eta_0 \eta_1(\by) \frac{\partial\eta_2(\by)}{\partial \widehat{\by}}\right)^{-1} \bigg(\bH_\gamma(\by) - \gamma \eta_0 \bigg[ \frac{\partial \eta_1(\by)}{\partial \by} \bigg] \eta_2(\by) \bigg). 
\end{equation}
The computation of $\partial \eta_2(\by)/\partial \widehat{\by}$ is analogous to that in the proof of Theorem \ref{th:df_glasso_ortho}. In order to compute the new part we observe that $\bW_{\Aset_G}$ is diagonal and relevant simplifications apply. Specifically, we have
\begin{equation*}
\bigg[ \frac{\partial \eta_1(\by)}{\partial \by} \bigg] \eta_2(\by) = \textrm{diag}(\Breve{\bbeta}_{\Aset_G})\frac{\partial \bW_{\Aset_G}(\by)}{\partial \by} = \sum_{g \in \Aset_G} \breve{\bbeta}_g \frac{\partial w_g(\by)}{\partial \by}.
\end{equation*}
The last step is to manipulate the gradient of weights through the chain rule, recalling the assumption $w_g(\by)=w_g(\lVert \widehat{\bbeta}_g^\mathsf{LS} \rVert_2)$.
\begin{equation*}
\frac{\partial w_g(\by)}{\partial \by} =  \frac{\partial w_g(\lVert \widehat{\bbeta}_g^\mathsf{LS} \rVert_2)}{\partial \by} = \frac{\partial w_g(\lVert \widehat{\bbeta}_g^\mathsf{LS} \rVert_2)}{\partial \lVert \widehat{\bbeta}_g^\mathsf{LS} \rVert_2} \frac{\partial  \lVert \widehat{\bbeta}_g^\mathsf{LS} \rVert_2}{(\partial \widehat{\bbeta}_g^\mathsf{LS})^{\T}}\frac{\partial  (\widehat{\bbeta}_g^\mathsf{LS})^{\T}}{\partial \by} = \frac{\partial w_g(\lVert \widehat{\bbeta}_g^\mathsf{LS} \rVert_2)}{\partial \lVert \widehat{\bbeta}_g^\mathsf{LS} \rVert_2} \frac{(\widehat{\bbeta}_g^\mathsf{LS})^{\T}}{\lVert\widehat{\bbeta}_g^\mathsf{LS}\rVert_2}\bX_g^{\T},
\end{equation*}
thus 
\begin{equation*}
\eta_0  \bigg[ \frac{\partial \eta_1(\by)}{\partial \by} \bigg] \eta_2(\by) = \bX_{\Aset_G} \bPhi_{\Aset_G} \bX_{\Aset_G}^{\T},
\end{equation*}
where 
\begin{equation*}
    \bPhi_{\Aset_G} = \text{blockdiag}\bigg( \frac{\widehat{\bbeta}_g}{\lVert \widehat{\bbeta}_g \rVert_2}\frac{\partial w_g(\lVert \widehat{\bbeta}_g^\mathsf{LS} \rVert_2)}{\partial \lVert \widehat{\bbeta}_g^\mathsf{LS} \rVert_2} \frac{(\widehat{\bbeta}_g^\mathsf{LS})^{\T}}{\lVert\widehat{\bbeta}_g^\mathsf{LS}\rVert_2}  \bigg)
\end{equation*}
To compute degrees of freedom, from Equation \eqref{eq:I_plus_eta2_inverse_H_minus_eta1} we write
\begin{equation*}
\begin{aligned}
\widehat{df}_\gamma &= \trace \bigg(\frac{\partial \widehat{\by}_\gamma}{\partial\by} \bigg) = \trace \bigg[\left(\bI_n+\gamma \eta_0 \eta_1(\by) \frac{\partial\eta_2(\by)}{\partial \widehat{\by}}\right)^{-1} \bigg(\bH_\gamma(\by) - \gamma \eta_0 \bigg[ \frac{\partial \eta_1(\by)}{\partial \by} \bigg] \eta_2(\by) \bigg) \bigg]\\
&= \trace \bigg[ \bigg( \bI_n + \gamma \bX_{\Aset_G} \bPi_{\Aset_G} \bX_{\Aset_G}^{\T} \bigg)^{-1} \bigg( \bX_{\Aset_G} \bX_{\Aset_G}^{\T} - \gamma \bX_{\Aset_G} \bPhi_{\Aset_g} \bX_{\Aset_G}^{\T}\bigg) \bigg],
\end{aligned}
\end{equation*}
completing the proof for the orthonormal case. For non-orthonormal designs we have $\bH_\gamma(\by)=\bX_{\Aset_G}(\bX_{\Aset_G}^{{\T}} \bX_{\Aset_G})^{-1} \bX_{\Aset_G}^{\T}$ and $
\eta_\gamma(\by)=\bX_{\Aset_G}(\bX_{\Aset_G}^{{\T}} \bX_{\Aset_G})^{-1} \bW_{\Aset_G}\Breve{\bbeta}_{\Aset_G}$. Previous computations can be done in a similar manner by considering $\eta_0 = \bX_{\Aset_G}(\bX_{\Aset_G}^{{\T}} \bX_{\Aset_G})^{-1}$, directly leading to the result.
\end{proof}
%
\begin{proof}[Proof of Corollary \ref{th:df_adaglasso_ortho}]
Under the setting of Theorem \ref{th:df_adaglasso_general_weights} and the orthonormal design assumption, we have that and $\widehat{\bbeta}^{\mathsf{LS}} = \bX^{\T} \by$ and the generic term of the matrix $\bW$ is given by $w_{gg} = w_g(\lVert\widehat{\bbeta}_g^\mathsf{LS}\rVert_2) =  1/\lVert\widehat{\bbeta}^{\mathsf{LS}}_g \rVert_2 =1/\lVert\bX_g^{\T} \by\rVert_2$. By applying the result of Theorem \ref{th:df_adaglasso_general_weights}
in the orthonormal case we have 
\begin{equation*}
\frac{\partial w_g(z)}{\partial z} \bigg|_{z=\widehat{\bbeta}_g^\mathsf{LS}} = \frac{\partial}{\partial z} \bigg( \frac{1}{z}\bigg)\bigg|_{z=\widehat{\bbeta}_g^\mathsf{LS}} =  - \frac{1}{z^2}\bigg|_{z=\widehat{\bbeta}_g^\mathsf{LS}} =  - \frac{1}{\lVert \widehat{\bbeta}_g^\mathsf{LS}\rVert_2^2},
\end{equation*}
which concludes the proof.
\end{proof}
%
\begin{proof}[Proof of Corollary \ref{th:df_adaglasso_nonortho}]
Under the setting of Theorem \ref{th:df_adaglasso_general_weights}, we have $\widehat{\bbeta}^{\mathsf{LS}} = (\bX^{\T}\bX)^{-1}\bX^{\T} \by$ and the generic term of the matrix $\bW$ is given by $w_{gg} = w_g(\lVert\widehat{\bbeta}_g^\mathsf{LS}\rVert_2) = 1/\lVert\widehat{\bbeta}^{\mathsf{LS}}_g\rVert_2 =1/\lVert\bp_g^{\T} \by\rVert_2$,
where $\bp_g\in\mathbb{R}^n$ is the matrix $\bP=\bX(\bX^{\T}\bX)^{-1}\in\mathbb{R}^{n\times p}$ restricted to the $g$-th group. By applying the result of Theorem \ref{th:df_adaglasso_general_weights} in the non-orthonormal case we have
\begin{equation*}
\frac{\partial w_g(z)}{\partial z} \bigg|_{z=\widehat{\bbeta}_g^\mathsf{LS}} = \frac{\partial}{\partial z} \bigg( \frac{1}{z}\bigg)\bigg|_{z=\widehat{\bbeta}_g^\mathsf{LS}} =  - \frac{1}{z^2}\bigg|_{z=\widehat{\bbeta}_g^\mathsf{LS}} =  - \frac{1}{\lVert\widehat{\bbeta}_g^\mathsf{LS}\rVert_2^2},
\end{equation*}
which concludes the proof.
\end{proof}
%
\begin{proof}[Proof of Corollary \ref{coro:bounds_df_adaglasso}]
We show that
\begin{align*}
\trace \big[\big( \bI_p + \gamma \bB\big)^{-1} \bA\big] &\leq \trace \big[\big( \bI_p + \gamma \bB\big)^{-1} \big(\bA - \gamma \bC \big)\big]=\trace \big[\big( \bI_p + \gamma \bB\big)^{-1} \bA \big] - \gamma \trace \big[\big( \bI_p + \gamma \bB\big)^{-1} \bC \big].
\end{align*}
Since $\gamma > 0$, $\bI_p$ and $\bB$ are positive definite, the previous inequality is true if $\bC$ is negative semidefinite, and thus if the quantity
\begin{equation*}
\bPhi_{\Aset_G} = \mathrm{blockdiag} \bigg(  \frac{\widehat{\bbeta}_g}{\lVert \widehat{\bbeta}_g \rVert_2 } \frac{\partial w_g(\lVert \widehat{\bbeta}^\mathsf{LS}\rVert_2)}{\partial \lVert\widehat{\bbeta}^\mathsf{LS} \rVert_2}\frac{(\widehat{\bbeta}^\mathsf{LS})^{\T}}{\lVert \widehat{\bbeta}^\mathsf{LS} \rVert_2}\bigg),
\end{equation*}
is negative semidefinite. By examination of this matrix, we conclude that both $\lVert\widehat{\bbeta}_g \rVert_2$ and $\lVert\widehat{\bbeta}_{g}^{\mathsf{LS}}\rVert_2$ are positive, $\partial w_g(\lVert \widehat{\bbeta}^\mathsf{LS}\rVert_2) / \partial \lVert\widehat{\bbeta}^\mathsf{LS} \rVert_2 $ is a negative scalar by assumption and the matrix $\widehat{\bbeta}_g(\widehat{\bbeta}_{g}^{\mathsf{LS}})^{\T}$ is positive semidefinite only if $(\widehat{\bbeta}_{g}^{\mathsf{LS}})^{\T}\widehat{\bbeta}_g \geq 0$.
\end{proof}

\end{appendix}

\bigskip 

\bibliographystyle{apalike} 
\bibliography{refs} 

%
%
%
%
%





\newpage
\renewcommand{\theequation}{S.\arabic{equation}}
\renewcommand{\thesection}{S.\arabic{section}}
\renewcommand{\thetable}{S.\arabic{table}}
\renewcommand{\theproposition}{S.\arabic{proposition}}
\renewcommand{\thetheorem}{S.\arabic{theorem}}
\renewcommand{\thecorollary}{S.\arabic{corollary}}

\setcounter{equation}{0}
\setcounter{table}{0}
\setcounter{section}{0}
\setcounter{proposition}{0}
\setcounter{corollary}{0}
\setcounter{theorem}{0}
\setcounter{page}{1}
\setcounter{footnote}{0}

\begin{center}
{\LARGE Supplementary materials for:
\vskip3mm
``Degrees of Freedom in Penalized Regression:\\ Model Selection with Adaptive Penalties''}
\vskip7mm
M. Bernardi$\null^{1}$, A. Canale$\null^{1}$ and M. Stefanucci$\null^{2}$
	\vskip5mm
	\centerline{\textit{$\null^1$Department of Statistics,
University of Padova}}
\vskip1mm
\centerline{\textit{ $\null^2$Department of Economics and Finance,
University of Rome Tor Vergata}}
\vskip6mm
\end{center}

\noindent These supplementary materials are organized as follows.  
Section~\ref{sec:appendix_additional_results} provides additional results together with their proofs, Section~\ref{sec:technical_app} collects several technical lemmas and auxiliary results used throughout the paper, and Section~\ref{sec:appendix_empirical_p>n} presents some empirical result in the $n<p$ setting.
%
%
\section{Additional results}
\label{sec:appendix_additional_results}
%
\begin{corollary}
\label{coro:df_adalasso_ortho_exp}
Let $\bX^{\T}\bX=\bI_p$, $\betahat$ the solution to the Adaptive Lasso problem in Equation \eqref{eq:convex_regularized_problem_adalasso} with weights equal to $w_j = \exp(-\alpha|\widehat{\beta}_j^{\mathsf{LS}}|)$ and $\gamma \in (\gamma_l, \gamma_{l+1})$. Denote with $\Aset$ the corresponding active set. An unbiased estimate of the degrees of freedom is
\begin{equation*}
\label{eq:dof_adalasso_ortho_exp}
\widehat{df}_\gamma = \vert\Aset\vert+\gamma\sum_{j\in\Aset} \frac{ \alpha}{\eexp(\alpha|\widehat{\beta}_j^{\mathsf{LS}}|)}.
\end{equation*}
\end{corollary}
\begin{proof}
Under the setting of Theorem \ref{th:df_adalasso_general_weights} and the orthonormal design assumption, we have that $\bH_\gamma(\by) = \bX_{\Aset} \bX_{\Aset}^{\T}$, $\eta_\gamma(\by) = \bX_{\Aset}\bW_{\Aset}(\by)\sgn(\betahat_{\Aset})$ and $\widehat{\bbeta}^{\mathsf{LS}} = \bX^{\T} \by$. Moreover, the matrix $\bW$ is given by
\begin{equation*}
\begin{aligned}
\bW=\mathrm{diag}\Big(\frac{1}{\eexp(\alpha\vert\betahat^{\mathsf{LS}}\vert)}\Big)=\mathrm{diag}\Big(\frac{1}{\eexp(\alpha \vert\bX^{\T}\by\vert)}\Big),
\end{aligned}
\end{equation*}
with generic term being $w_{jj} = w_j(|\widehat{\beta}_j^\mathsf{LS}|) =  1/\eexp(\alpha|\widehat{\beta}^{\mathsf{LS}}_j|) =1/\eexp(\alpha|\bx_j^{\T} \by|)$. By applying the result of Theorem \ref{th:df_adalasso_general_weights}
in the orthonormal case we have 
\begin{equation*}
\frac{\partial w_j(z)}{\partial z} \bigg|_{z=\widehat{\beta}_j^\mathsf{LS}} = \frac{\partial}{\partial z} \bigg( \frac{1}{\eexp(\alpha z)}\bigg)\bigg|_{z=\widehat{\beta}_j^\mathsf{LS}} =  - \frac{\alpha}{\eexp(\alpha z)}\bigg|_{z=\widehat{\beta}_j^\mathsf{LS}} =  - \frac{\alpha}{\eexp(\alpha|\widehat{\beta}_j^\mathsf{LS}|)},
\end{equation*}
which concludes the proof.
\end{proof}
%
\begin{corollary}
\label{coro:df_adalasso_nonortho_exp}
Let $\betahat$ be the solution to the Adaptive Lasso problem in Equation \eqref{eq:convex_regularized_problem_adalasso} with weights equal to $w_j = \exp(-\alpha|\widehat{\beta}_j^{\mathsf{LS}}|)$ and $\gamma \in (\gamma_l, \gamma_{l+1})$. Denote with $\Aset$ the corresponding active set. An unbiased estimate of the degrees of freedom is
\begin{equation*} 
\label{eq:dof_adalasso_nonortho_exp}
\widehat{df}_\gamma=\vert\Aset\vert+\gamma\alpha\sum_{j\in\Aset} \sgn(\widehat{\beta}_{j}) \frac{ \sgn(\widehat{\beta}_j^{\mathsf{LS}})}{\eexp(\alpha|\widehat{\beta}_j^{\mathsf{LS}}|)} [(\bX_{\Aset}^{\T} \bX_{\Aset} )^{-1}]_{\pi(j),\pi(j)}.
\end{equation*}
\end{corollary}
\begin{proof}
Under the setting of Theorem \ref{th:df_adalasso_general_weights}, we have $\widehat{\bbeta}^{\mathsf{LS}} = (\bX^{\T}\bX)^{-1}\bX^{\T} \by$ and
\begin{equation*}
\begin{aligned}
\bW=\mathrm{diag}\Big(\frac{1}{\eexp(\alpha\vert\betahat^{\mathsf{LS}}\vert)}\Big)=\mathrm{diag}\Big(\frac{1}{\eexp(\alpha\vert(\bX^{\T}\bX)^{-1}\bX^{\T}\by\vert)}\Big),
\end{aligned}
\end{equation*}
with generic term being $w_{jj} = w_j(|\widehat{\beta}_j^\mathsf{LS}|) = 1/\eexp(\alpha|\widehat{\beta}^{\mathsf{LS}}_j|) =1/\eexp(\alpha|\bp_j^{\T} \by|)$, where $\bp_j\in\mathbb{R}^n$ is the $j$-th column of $\bP=\bX(\bX^{\T}\bX)^{-1}\in\mathbb{R}^{n\times p}$. By applying the result of Theorem \ref{th:df_adalasso_general_weights} in the non-orthonormal case we have
\begin{equation*}
\frac{\partial w_j(z)}{\partial z} \bigg|_{z=\widehat{\beta}_j^\mathsf{LS}} = \frac{\partial}{\partial z} \bigg( \frac{1}{\eexp(\alpha z)}\bigg)\bigg|_{z=\widehat{\beta}_j^\mathsf{LS}} =  - \frac{\alpha}{\eexp(\alpha z)}\bigg|_{z=\widehat{\beta}_j^\mathsf{LS}} =  - \frac{\alpha}{\eexp(\alpha|\widehat{\beta}_j^\mathsf{LS}|)},
\end{equation*}
which concludes the proof.
\end{proof}
%
\begin{corollary}
\label{coro:similarexpressions}
Under the settings of Theorem \ref{th:df_glasso_ortho}, the following are equivalent 
\begin{align*}
\label{eq:dof_glasso_ortho_simple}
&\widehat{df}_\gamma = |\Aset_G| + \sum_{g \in \Aset_G} \frac{n_g-1}{1+\gamma\frac{w_g}{\lVert \bbeta_g\rVert_2}}  \\
&\widehat{df}_\gamma = |\Aset_p| - \sum_{g \in \Aset_G} (n_g-1)\frac{\gamma w_g / \lVert \bbeta_g \rVert_2}{1+\gamma w_g / \lVert \bbeta_g \rVert_2}.
\end{align*}
\end{corollary}
\begin{proof}
The eigenvalues of $\bA$ and $\bB$ are, respectively,  
$$ \lambda(\bA)= \big(\underbrace{1, \ldots, 1}_{|\Aset_p|}, \underbrace{0, \ldots, 0}_{n-|\Aset_p|}\big), \quad \lambda(\bB) = \bigg(\bigcup_{g \in \Aset} \bigg\{ \underbrace{\frac{w_g}{\lVert \widehat{\bbeta}_g \rVert_2}, \ldots, \frac{w_g}{\lVert \widehat{\bbeta}_g \rVert_2}}_{n_g-1}, 0 \bigg\}, \underbrace{0, \ldots, 0 }_{n-|\Aset_p|}\bigg),$$
thus $\bB$ has $|\Aset_p| - |\Aset_G|$ eigenvalues greater than zero. We have that 
\begin{align*}
&\trace ((\bI_n+\gamma\bB)^{-1}\bA) =\sum_{i=1}^n \frac{\lambda_i^A}{1+\gamma\lambda_i^B} = \sum_{g \in \Aset_G} \sum_{j=1}^{n_g} \frac{\lambda_{n_{g-1}+j}^A}{1+\gamma\lambda_{n_{g-1}+j}^B} = \sum_{g \in \Aset_G} \bigg\{\sum_{j=1}^{n_g-1} \frac{1}{1+\gamma \frac{w_g}{\lVert \widehat{\bbeta}_g \rVert_2}} + 1 \bigg\}\\ 
&=|\Aset_G| + \sum_{g \in \Aset_G}\sum_{j =1}^{n_g-1} \frac{1}{1+\gamma \frac{w_g}{\lVert \widehat{\bbeta}_g \rVert_2}}=
|\Aset_G| + \sum_{g \in \Aset_G} \frac{n_g-1}{1+\gamma\frac{w_g}{\lVert \widehat{\bbeta}_g\rVert_2} },
\end{align*}
which is the result presented in \cite{yuan:2006:glasso}. By contrast,
\begin{align*}
\trace ((\bI_n+\gamma\bB)^{-1}\bA) &= \trace (\bA - \gamma\bB(\bI_n+\gamma\bB)^{-1}\bA)\\ &= \trace (\bA) - \trace (\gamma\bB(\bI_n+\gamma\bB)^{-1}\bA)\\
&=|\Aset_p| - \sum_{i=1}^n \frac{\gamma \lambda_i^B}{1+\gamma\lambda_i^B} \lambda_i^A = |\Aset_p| - \sum_{g \in \Aset_G} \sum_{j=1}^{n_g} \frac{\gamma \lambda_{n_{g-1}+j}^B}{1+\gamma\lambda_{n_{g-1}+j}^B}\lambda_{n_{g-1}+j}^A\\  
&=|\Aset_p| - \sum_{g \in \Aset_G} \sum_{j=1}^{n_g-1} \frac{\gamma w_g / \lVert \widehat{\bbeta}_g \rVert_2}{1+\gamma w_g / \lVert \widehat{\bbeta}_g \rVert_2}= |\Aset_p| - \sum_{g \in \Aset_G} (n_g-1)\frac{\gamma w_g / \lVert \widehat{\bbeta}_g \rVert_2}{1+\gamma w_g / \lVert \widehat{\bbeta}_g \rVert_2},
\end{align*}
which is the result presented in \cite{vaiter_etal.2017}. Similar simplifications under the non orthonormal setting are not easy to derive, as this proof is based on the eigenvalues of the matrix $\bB$ which are unknown unless $\bX^{\T}\bX=\bI_p$. 
\end{proof}
%
\begin{corollary}
\label{coro:group1}
Under the settings of Theorems \ref{th:df_glasso_ortho} and \ref{th:df_glasso_nonortho}, if $n_g=1$ for all $g \in \mathcal{G}_G$ we have $\widehat{df} = |\Aset|$, recovering the Lasso result. 
\end{corollary}
\begin{proof}
In such circumstances the problem in Equation \eqref{eq:convex_regularized_problem_glasso} is equivalent to the problem in Equation \eqref{eq:convex_regularized_problem_adalasso} with fixed weights, $G=p$, $\lVert \bbeta_g \rVert_2 = |\beta_g|$ and $|\Aset_G| = |\Aset_p|$. If $\bX^{\T}\bX=\bI_p$, starting from Corollary 3 we have that $n_g=1$, and the result is $\widehat{df} = \lvert \Aset \rvert$, which is the number of active groups and active variables at the same time. For the non-orthogonal case, it turns out that $\lVert\bbeta_g\rVert_2^2 = \bbeta_g^{\T}\bbeta_g = \bbeta_g \bbeta_g^{\T}$ and $\bPi_\Aset$ is null, and we have again that $\widehat{df} = \trace(\bA) = \trace(\bH_\gamma) = \lvert \Aset \rvert$. 
\end{proof}
%
\begin{corollary}
\label{coro:dbeta_dgamma_orto}
Under the setting of Theorem \ref{th:monotonicity_df_glasso}, for the orthonormal design, $\bX^{\T}\bX=\bI_p$, we have:
\begin{equation*}
\frac{\mathrm{d}\betahat_g}{\mathrm{d}\gamma} = - \frac{w_g}{r_g}  \betahat_g, \qquad  \frac{{\mathrm d}r_{g}}{{\mathrm d} \gamma} = -w_g, \qquad \frac{{\mathrm d}\bPi_{g}}{{\mathrm d} \gamma} = \frac{w_g}{r_g} \bPi_{g} , \qquad \text{for each } g \in \Aset_G,
\end{equation*}
where $r_g= \Vert\betahat_g\Vert_2$, which means that the solution scales linearly in the direction of $\betahat_g$ as $\gamma$ increases, for each active group.
\end{corollary}
\begin{proof}
Let us compute the expression for the derivative of $\betahat_g$ with respect to the regularization parameter $\gamma$, when $\bX^{\T}\bX=\bI_p$. Equation \eqref{eq:IFT_Pi_FINAL} reduces to the following expression:
\begin{equation}
\label{eq:betahat_der_orto}
\frac{\mathrm{d}\betahat_g}{\mathrm{d}\gamma}
= -\frac{w_g}{r_g}( \bI_{n_g} + \gamma \bPi_g)^{-1} \betahat_g,
\end{equation}
where $r_g= \Vert\betahat_g\Vert_2$ and $\bPi_g$ has been defined in Equation \eqref{eq:ABPi_g_definition}.
Since $\boldsymbol{\Pi}_g$ is symmetric, it can be diagonalized. Consider the eigensystem of $\boldsymbol{\Pi}_g$. The direction $\betahat_g$ is an eigenvector of $\boldsymbol{\Pi}_g$. Let $\bu_g = \betahat_g / \Vert\betahat_g\Vert_2 = \betahat_g / r_g$ be a unit vector, then:
\begin{equation*}
\begin{aligned}
\boldsymbol{\Pi}_g \bu_g &= w_g\left( \frac{\bI_{n_g}}{r_g} - \frac{\betahat_g \betahat_g^{\T}}{r_g^3} \right) \bu_g
= w_g\frac{\bu_g}{r_g} - w_g\frac{\betahat_g \betahat_g^{\T} \bu_g}{r_g^3}\\
&= w_g\frac{\bu_g}{r_g} - w_g\frac{\betahat_g  (\betahat_g^{\T} \bu_g)}{r_g^3}
= w_g\frac{\bu_g}{r_g} - w_g\frac{\betahat_g  r_g}{r_g^3}
= w_g\frac{\bu_g}{r_g} - w_g\frac{r_g^2 \bu_g}{r_g^3} = 0,
\end{aligned}
\end{equation*}
therefore, $\bu_g$ is an eigenvector with eigenvalue equal to $0$. Moreover, for any vector $\bv_g \perp \bu_g$, we have:
\begin{equation*}
\boldsymbol{\Pi} \bv_g = w_g\left( \frac{\bI_{n_g}}{r_g} - \frac{\betahat_g \betahat_g^{\T}}{r_g^3} \right) \bv_g
= w_g\frac{\bv_g}{r} - w_g\frac{\betahat_g (\betahat_g^{\T} \bv_g)}{r_g^3}
= \frac{\bv_g}{r_g},  
\end{equation*}
so every orthogonal direction to $\betahat_g$ has eigenvalue equal to $w_g/r_g$. Thus, the eigenvalues of $\boldsymbol{\Pi}_g$ are $\lambda_1=0$ with multiplicity equal to one, and $\lambda_{n_g} = w_g/r_g$ with multiplicity equal to $n_g-1$. Therefore, along $\betahat_g$ the eigenvalue of $\bI_{n_g} + \gamma\boldsymbol{\Pi}_g$ is $1$, while orthogonal to $\betahat_g$ the eigenvalues of $\bI_{n_g} + \gamma\boldsymbol{\Pi}_g$ are $1 + \gamma w_g/r_g$. Let us now compute $\mathrm{d}\betahat_g/\mathrm{d}\gamma$ in Equation \eqref{eq:betahat_der_orto}. Since $\bI_{n_g} + \gamma \bPi_g$ acts as identity on $\betahat_g$,the same is true for its inverse, i.e. 
$ (\bI_{n_g} + \gamma\boldsymbol{\Pi}_g)^{-1} \betahat_g =  \betahat_g$,
and thus 
\begin{equation*}
\frac{\mathrm{d}\betahat_g}{\mathrm{d}\gamma} = - w_g  \frac{\betahat_g}{r_g},
\end{equation*}
which completes the first part of the proof. Equation \eqref{eq:dr_gdgamma} in the orthonormal setting reduces to
\begin{equation*}
\frac{\mathrm{d}r_g}{\mathrm{d}\gamma}
= \frac{\mathrm{d}}{\mathrm{d}\gamma} \left( \Vert \widehat{\bbeta}_{g} \Vert_2 \right)
=-\frac{1}{r_g}\widehat{\bbeta}_{g}^{\T} \Big[\bR_{\Aset_G}^\star\betahat_{\Aset_G}\Big]_g = -\frac{w_g}{r^2_g}\widehat{\bbeta}_{g}^{\T} ( \bI_{n_g} + \gamma \bPi_g)^{-1} \betahat_g = -\frac{w_g}{r^2_g}\widehat{\bbeta}_{g}^{\T} \betahat_g = -w_g
\end{equation*}
which completes the second part of the proof. In order to prove the last expression, let us consider Equation \eqref{eq:Pi_derivative_final} in the orthonormal setting:
\begin{align*}
\frac{\mathrm{d}\bPi_{g}}{\mathrm{d}\gamma}
&= -\frac{w_g}{r_g^2}\frac{\mathrm{d}r_g}{\mathrm{d}\gamma}\Big(\bI_{n_g}-\frac{3}{r_g^2} \widehat{\bbeta}_g \widehat{\bbeta}_g^{\T} \Big)+  \frac{w_g}{r_g^3} \Big(\Big[\bR_{\Aset_G}^\star\widehat{\bbeta}_{\Aset_G}\Big]_g \widehat{\bbeta}_g^{\T} + \widehat{\bbeta}_g \Big[\bR_{\Aset_G}^\star\widehat{\bbeta}_{\Aset_G}\Big]_g^{\T} \Big)\\
&= \frac{w_g^2}{r_g^2}\Big(\bI_{n_g}-\frac{3}{r_g^2} \widehat{\bbeta}_g \widehat{\bbeta}_g^{\T} \Big)+ \frac{w_g^2}{r_g^4} \Big(\widehat{\bbeta}_g \widehat{\bbeta}_g^{\T} + \widehat{\bbeta}_g \widehat{\bbeta}_g^{\T}  \Big)\\
&=  \frac{w_g^2}{r_g^2}\Big(\bI_{n_g}-\frac{3}{r_g^2} \widehat{\bbeta}_g \widehat{\bbeta}_g^{\T} \Big)+ 2\frac{w_g^2}{r_g^4} \widehat{\bbeta}_g \widehat{\bbeta}_g^{\T}
=  \frac{w_g^2}{r_g^2}\Big(\bI_{n_g}-\frac{1}{r_g^2} \widehat{\bbeta}_g \widehat{\bbeta}_g^{\T} \Big)= \frac{w_g}{r_g}\bPi_g.
\end{align*}
\end{proof}
%
\begin{corollary}
\label{coro:alternative_adaglasso}
Under the settings of Corollary \ref{th:df_adaglasso_ortho},  the following is equivalent to Equation \eqref{eq:dof_adaglasso_ortho}.
\begin{equation*}
\widehat{df}_\gamma = |\Aset_G| + \sum_{g \in \Aset_G} \bigg(n_g-1+\frac{\gamma}{\lVert\widehat{\bbeta}_g^\mathsf{LS}\rVert_2^2}\bigg) \frac{1}{1+\gamma\frac{w_g}{\lVert \widehat{\bbeta}_g\rVert_2}}
\end{equation*}
\end{corollary}
\begin{proof}
Starting from Equation \eqref{eq:dof_adaglasso_ortho} we have
\begin{equation*}
\widehat{df}_\gamma = \trace \big[ \big(\bI_n+\gamma \bB \big)^{-1} \bA \big] + \gamma \trace \big[\big(\bI_n+\gamma \bB \big)^{-1} \bC \big] 
\end{equation*}
About the former quantity, as in Group Lasso, we have
\begin{equation*}
\trace \big[ \big(\bI_n+\gamma \bB \big)^{-1} \bA \big] = |\Aset_G| + \sum_{g \in \Aset_G} \frac{n_g-1}{1+\gamma\frac{w_g}{\lVert \widehat{\bbeta}_g\rVert_2} }.
\end{equation*}
For the latter, note that matrix $\bC$ is positive semidefinite and, since for the orthonormal design we have $\widehat{\bbeta}_g = c \widehat{\bbeta}^\mathsf{LS}$ for certain $c>0$, its spectrum is 
\begin{equation*}
\lambda(\bC) = \bigg(\bigcup_{g \in \Aset} \bigg\{ \frac{1}{\lVert \widehat{\bbeta}_g^\mathsf{LS} \rVert_2^2}, \underbrace{0, \ldots, 0}_{n_g-1} \bigg\}, \underbrace{0, \ldots, 0 }_{n-|\Aset_p|}\bigg) = \bigg(\bigcup_{g \in \Aset} \bigg\{ \frac{1}{\lVert \widehat{\bbeta}_g^\mathsf{LS} \rVert_2^2}\bigg\}, \underbrace{0, \ldots, 0 }_{n-|\Aset_G|}\bigg).
\end{equation*}
Thus, 
\begin{equation*}
\trace \big[ \big(\bI_n+\gamma \bB \big)^{-1} \bC \big] = \sum_{i=1}^n \frac{\lambda_i^C}{1+\lambda_i^B}= \sum_{g \in \Aset_G} \frac{1/\lVert \widehat{\bbeta}_g^\mathsf{LS}\rVert_2^2}{1+\gamma\frac{w_g}{\lVert \widehat{\bbeta}_g\rVert_2} }.
\end{equation*}
which leads directly to the result.
\end{proof}
%
\begin{corollary}
\label{coro:special_cases_adaglasso}
Under the settings of Theorems \ref{th:df_adaglasso_ortho} and \ref{th:df_adaglasso_nonortho}, we recover the Adaptive Lasso as a limiting case when $n_g=1$ for all $g \in \mathcal{G}_G$, the Group Lasso if weights $w_g(\by) = w_g$ do not depend on $\by$, and the Lasso if both these conditions are satisfied.
\end{corollary}
\begin{proof}
When $n_g=1$ the problem in Equation \eqref{eq:convex_regularized_problem_adaglasso} is equivalent to the problem in Equation \eqref{eq:convex_regularized_problem_adalasso} with adaptive weights, $G=p$, $\lVert \bbeta_g \rVert_2 = |\beta_g|$ and $|\Aset_G| = |\Aset_p|$. It turns out that $\lVert\bbeta_g\rVert_2^2 = \bbeta_g^{\T}\bbeta_g = \bbeta_g \bbeta_g^{\T}$ and $\bPi_\Aset$ is null. Moreover, $ \lVert \widehat{\bbeta}_g \rVert_2 = | \widehat{\bbeta}_g |_1$ and $ \lVert \widehat{\bbeta}_g \rVert_2^3 = | \widehat{\bbeta}_g |_1^3 = \widehat{\bbeta}_g^2 |\widehat{\bbeta}_g|_1$ and 
\begin{equation*}
\bPhi_{\Aset_g} = \bPhi_{j} = \frac{\widehat{\beta}_j^{\mathsf{LS}} \widehat{\beta}_j}{ |\widehat{\beta}_j^{\mathsf{LS}}| |\widehat{\beta}_j|}\frac{\partial w_j(z)}{\partial z} \bigg|_{z = \widehat{\beta}_j^\mathsf{LS}} = \text{sgn}(\widehat{\beta}_j^{\mathsf{LS}})\text{sgn}(\widehat{\beta}_j)\frac{\partial w_j(z)}{\partial z} \bigg|_{z = \widehat{\beta}_j^\mathsf{LS}},
\end{equation*}
which is exactly the quantity that we introduced in Theorem \ref{th:df_adalasso_general_weights}. When the weights do not depend on the response, we have $\partial w_g / \partial \by = 0$ and thus $\bPhi_{\Aset} = \boldsymbol{0}$ and $\bC = \boldsymbol{0}$ leading directly to Theorem \ref{th:df_glasso_ortho} and \ref{th:df_glasso_nonortho}. Finally, if $n_g=1$ and the weights are independent from $\by$, both simplifications apply, leading to standard Lasso result.
\end{proof}

\section{Technical Results}
\label{sec:technical_app}
\begin{lemma}
\label{lemma:adalasso_weights_derivative}
Let $\bA\in\mathbb{R}^{n\times p}$, $\bb\in\mathbb{R}^{p}$ and $\bW_y=\diag(f_\varrho(\bw_y))\in\mathbb{S}^{p\times p}$, where $\bw_y=\bP^{\T}\by\in\mathbb{R}^{p}$, with $\bP=\bX(\bX^{\T}\bX)^{-1}$, $\bX\in\mathbb{R}^{n\times p}$, $\by\in\mathbb{R}^n$ and $f_\varrho(x):\mathbb{R}\rightarrow\mathbb{R}^+$ is a positive function that applies element-wise and $\varrho\in\mathbb{R}$ is an additional parameter controlling the weighting function, then
\begin{equation*}
\frac{\partial(\bA\bW_y\bb)}{\partial\by}=\sum_{g=1}^p b_g f_\varrho^{\prime}(\bp_g^{\T}\by) \bp_g\ba_g^{{\T}},
\end{equation*}
where $f_\varrho^\prime(x):\mathbb{R}\rightarrow\mathbb{R}$ is the first derivative of $f_\varrho(x)$, $w_{y,g}$ denotes the $g$-th diagonal element of the matrix $\bW_y$, $\bp_g$ and $\ba_g$ denote the $g$-th column of the matrix $\bP$, and $\bA$ respectively, and $\be_g$ is the $g$-th column vector of $\bI_p$. 
If $f(x)=\vert x\vert^{-1}$, then $\frac{\partial f_\varrho(x)}{\partial x}\Big\vert_{\bp_g^{\T}\by}=\frac{-\sgn(\bp_g^{\T}\by)}{(\bp_g^{\T}\by)^2}$, while if $f(x)=\exp{(-\varrho\vert x\vert)}$ then $\frac{\partial f_\varrho(x)}{\partial x}\Big\vert_{\bp_g^{\T}\by}=-\frac{\varrho\sgn(\bp_g^{\T}\by)}{\exp{(\varrho\bp_g^{\T}\by)}}$. 
\end{lemma}
\begin{proof}
By exploiting the properties of the vectorization operator, as detailed in \cite{magnus_neudecker.1999}, we get:
\begin{equation*}
\label{eq:dof_ovglasso_weights_function_fd_prima}
\frac{\partial(\bA\bW_y\bb)}{\partial\by}=\frac{\partial\vecof(\bA\bW_y\bb)}{\partial\by}=(\bb^{\T}\otimes\bA)\frac{\partial\vecof(\bW_y)}{\partial\by},
\end{equation*}
where
\begin{equation}
\label{eq:dof_ovglasso_weights_function_fd}
\frac{\partial\vecof(\bW_y)}{\partial\by}=
\begin{bmatrix}
\be_1  \big(\frac{\partial w_{y,1}}{\partial \by}\big)^{\T}\\
\be_2 \big(\frac{\partial w_{y,2}}{\partial \by}\big)^{\T}\\
\vdots\\
\be_p  \big(\frac{\partial w_{y,p}}{\partial \by}\big)^{\T}
\end{bmatrix}\in\mathbb{R}^{p^2\times n},
\end{equation}
and using the chain rule for derivatives of composite functions
\begin{equation}
\label{eq:dof_ovglasso_weights_function_fd_dopo}
\frac{\partial w_{y,g}}{\partial \by}=f^\prime\big(\bp_g^{\T}\by\big)\frac{\partial \bp_g^{\T}\by}{\partial \by}=f^\prime\big(\bp_g^{\T}\by\big) \bp_g\in\mathbb{R}^n,
\end{equation}
for all $g=1,\dots,p$. Equations \eqref{eq:dof_ovglasso_weights_function_fd}-\eqref{eq:dof_ovglasso_weights_function_fd_dopo} further simplify to
\begin{equation*}
\frac{\partial(\bA\bW_y\bb)}{\partial\by}=\sum_{g=1}^p b_g \frac{\partial w_{y,g}}{\partial\by}  \ba_g^{{\T}}
=\sum_{g=1}^p b_g f^\prime\big(\bp_g^{\T}\by\big) \bp_g\ba_g^{{\T}}.
\end{equation*}
If $f(x)=\vert x\vert^{-1}$, then $f^\prime(x)=-\frac{\sgn(x)}{x^{2}}$ and $\frac{\partial w_{y,g}}{\partial \by}=\frac{-\sgn(\bp_g^{\T}\by)}{(\bp_g^{\T}\by)^2}\bp_g$. If $f(x)=\exp{(-\varrho\vert x\vert)}$ then $f^\prime(x)=-\varrho f(x)\sgn(x)$ and $\frac{\partial w_{y,g}}{\partial \by}=-\frac{\varrho\sgn(\bp_g^{\T}\by)}{\exp{(\varrho\bp_g^{\T}\by)}}\bp_g$. This completes the proof.
\end{proof}

%
\begin{lemma}
\label{lemma:product_Ax_norm_fd}
Let $\bA\in\mathbb{R}^{n\times p}$ and $\bx\in\mathbb{R}^p$, such that $\Vert\bA\bx\Vert_2 \neq 0$, then
\begin{equation*}
\frac{\partial}{\partial \bx} \Bigg(\frac{\bA^{{\T}} \bA \bx}{\Vert\bA\bx\Vert_2}\Bigg)=\frac{\bA^{{\T}} \bA}{\Vert\bA \bx\Vert_2}-\frac{\bA^{{\T}} \bA \bx\bx^{{\T}} \bA^{{\T}} \bA}{\Vert\bA \bx\Vert_2^3}.
\end{equation*}  
\end{lemma}
\begin{proof}
Applying the rule for the derivative of a quotient of two differentiable functions, we get:
\begin{equation*}
\begin{aligned}
\frac{\partial}{\partial \bx} \Bigg(\frac{\bA^{{\T}} \bA \bx}{\Vert\bA\bx\Vert_2}\Bigg)&= \frac{\bA^{{\T}} \bA\Vert \bA \bx\Vert_2}{\Vert\bA \bx\Vert_2^2}-\frac{\bA^{{\T}} \bA \bx}{\Vert\bA\bx\Vert_2^2}\frac{\partial}{\partial \bx} \sqrt{\Vert\bA \bx\Vert_2^2} \\
&= \frac{\bA^{{\T}} \bA}{\Vert\bA \bx\Vert_2}-\frac{\bA^{{\T}} \bA \bx}{\Vert\bA \bx\Vert_2^2}  \frac{1}{2\Vert\bA\bx\Vert_2} 2(\bA^{\T}\bA\bx)^{\T},
\end{aligned}
\end{equation*}
which completes the proof.
\end{proof}

\begin{lemma}
\label{lemma:Pi_eigenvalues_all}
Let $\boldsymbol{u},\boldsymbol{v} \in \mathbb{R}^n$ with $n\geq 2$, $\bu\neq\bv$, $a,b,c\in\mathbb{R}\setminus \{0\}$, then the spectrum of the matrix:
\begin{equation*}
\boldsymbol{\Pi} = a \boldsymbol{I}_n + b \boldsymbol{u} \boldsymbol{u}^{\T} + c(\boldsymbol{u}\boldsymbol{v}^{\T} + \boldsymbol{v}\boldsymbol{u}^{\T})\in\mathbb{S}^n,
\end{equation*}
is
\begin{align*}
\lambda_1 &= \frac{1}{2}(2 a+b \Vert\bu\Vert_2^2+2 c \bv^{\T}\bu) + \frac{1}{2}\sqrt{\Delta}, \\
\lambda_2 &= a, \qquad \text{with multiplicity } n - 2, \\
\lambda_3 &= \frac{1}{2}(2 a+b \Vert\bu\Vert_2^2+2 c \bv^{\T}\bu) - \frac{1}{2}\sqrt{\Delta},
\end{align*}
and $\Delta =4c^2\Vert\bu\Vert_2^2\Vert\bv\Vert_2^2+b^2\Vert\bu\Vert_2^4+4bc\Vert\bu\Vert_2^2\bv^{\T}\bu$, if $n>2$ and $(\lambda_1,\lambda_3)$ if $n=2$.
%
\end{lemma}
\begin{proof}
Let $\boldsymbol{M}=b \boldsymbol{u} \boldsymbol{u}^{{\T}}+c\left(\boldsymbol{u} \boldsymbol{v}^{{\T}}+\boldsymbol{v} \boldsymbol{u}^{{\T}}\right)\in\mathbb{S}^n$, then $\mathrm{rank}(\boldsymbol{M}) \leq \mathrm{rank}\left(b \boldsymbol{u} \boldsymbol{u}^{{\T}}\right)+\mathrm{rank}\left(c\left(\boldsymbol{u} \boldsymbol{v}^{{\T}}+\boldsymbol{v} \boldsymbol{u}^{{\T}}\right)\right)$ $ \leq 3$. This means that the image of $\bM$ lies in $\mathcal{S} = \text{span}\{\boldsymbol{u}, \boldsymbol{v}\} \subset \mathbb{R}^n$, which is a two-dimensional subspace. Moreover, adding $a\bI_n$ simply shifts the eigenvalues and does not affect the rank, therefore $\mathrm{rank}(\bPi)\leq 3$.
For any $ \boldsymbol{w} \in \mathcal{S}^\perp$, we have:
\begin{equation*}
\boldsymbol{\Pi} \boldsymbol{w} = a \boldsymbol{w},
\end{equation*}
from which we find that $\lambda_2 = a$ is an eigenvalue of $\boldsymbol{\Pi}$, 
with multiplicity equal to $n-2$. We now compute the restriction of $\boldsymbol{\Pi}$ to the 2-dimensional subspace $\mathcal{S}$. Let:
\begin{equation*}
\alpha = \frac{\boldsymbol{u}^{\T} \boldsymbol{v}}{\Vert\bu\Vert_2}, \qquad \beta =\sqrt{\Vert\bv\Vert_2^2 - \alpha^2},
\end{equation*}
we define an orthonormal basis $\{\be_1, \be_2\}$ for $\mathcal{S}$ as $\be_1 = \frac{\bu}{\Vert\bu\Vert_2} $ and decompose $\bv$ as $\bv = \alpha\be_1 + \beta \be_2$, where $\be_2 = \frac{\bv - \alpha \be_1}{\beta}$.
We express $\boldsymbol{\Pi}$ on the basis $ \{\be_1, \be_2\} $, and we rewrite $ \boldsymbol{\Pi} $ as:
\begin{equation*}
\boldsymbol{\Pi} =a\boldsymbol{I}_n+b\Vert \bu\Vert_2^2\be_1\be_1^{\T} +c\Vert \bu\Vert_2 ( \be_1 \bv^{\T} + \bv \be_1^{\T}),
\end{equation*}
and substituting the expression for $\bv = \alpha \be_1 + \beta \be_2 $, we get:
\begin{equation*}
\begin{aligned}
\boldsymbol{\Pi} &= a \boldsymbol{I}_n+b\Vert \bu\Vert_2^2\be_1\be_1^{\T}+c\Vert \bu\Vert_2 \left[\be_1 (\alpha \be_1^{\T} + \beta \be_2^{\T}) + (\alpha \be_1 + \beta \be_2) \be_1^{\T} \right]\\
&= a\boldsymbol{I}_n+br^2\be_1\be_1^{\T}+cr \left[ 2\alpha \be_1 \be_1^{\T} + \beta (\be_1 \be_2^{\T} + \be_2 \be_1^{\T}) \right],
\end{aligned}
\end{equation*}
where $r = \Vert\boldsymbol{u}\Vert_2$.
Therefore, on the orthonormal basis $\{\be_1, \be_2\}$, the restriction of $ \boldsymbol{\Pi} $ to $ \mathcal{S}$ has the matrix:
\[
\boldsymbol{\Pi}_{\mathcal{S}} =
\begin{bmatrix}
a+br^2+2cr\alpha & cr\beta \\
cr\beta & a
\end{bmatrix}.
\]
To compute its eigenvalues of $\boldsymbol{\Pi}_{\mathcal{S}}$, we solve the characteristic polynomial:
$$
\lambda^2-\mathrm{trace}\left(\boldsymbol{\Pi}_{\mathcal{S}}\right) \lambda+\vert\boldsymbol{\Pi}_{\mathcal{S}}\vert=0,
$$
where
$$
\begin{aligned}
\mathrm{trace}\left(\boldsymbol{\Pi}_{\mathcal{S}}\right)&=\left(a+b r^2+2 c r \alpha\right)+a=2 a+b r^2+2 c r \alpha\\
\vert\boldsymbol{\Pi}_{\mathcal{S}}\vert&=\left(a+b r^2+2 c r \alpha\right)a-(c r \beta)^2=a\left(a+b r^2+2 c r \alpha\right)-c^2 r^2 \beta^2,
\end{aligned}
$$
and therefore
$$
\lambda_{1,2}=\frac{1}{2}\left[\mathrm{trace}\left(\boldsymbol{\Pi}_{\mathcal{S}}\right) \pm \sqrt{\mathrm{trace}\left(\boldsymbol{\Pi}_{\mathcal{S}}\right)^2-4 \vert\boldsymbol{\Pi}_{\mathcal{S}}\vert}\right].
$$
The complete set of eigenvalues of $\boldsymbol{\Pi}$ is:
\begin{align*}
\lambda_1 &= \frac{1}{2}(2 a+b r^2+2 c r \alpha) + \frac{1}{2}\sqrt{\Delta}, \\
\lambda_2 &= a, \qquad \text{with multiplicity } n - 2, \\
\lambda_3 &= \frac{1}{2}(2 a+b r^2+2 c r \alpha) - \frac{1}{2}\sqrt{\Delta},
\end{align*}
where $\Delta=\mathrm{trace}\left(\boldsymbol{\Pi}_{\mathcal{S}}\right)^2-4 \vert\boldsymbol{\Pi}_{\mathcal{S}}\vert=4c^2r^2(\alpha^2+\beta^2)+b^2r^4+4bcr^3\alpha$. Substituting for the expression of $r$, $\alpha$ and $\beta$ completes the proof. Concerning the case $n=2$, the complete set of eigenvalues is $(\lambda_1, \lambda_3)$ since $\mathcal{S}^\perp = \emptyset$.
\end{proof}

\begin{lemma}
\label{lemma:Delta_sign}
Let $\boldsymbol{u},\boldsymbol{v} \in \mathbb{R}^n$ with $n\geq 2$, $\gamma, w, r \in \mathbb{R}_+$, then the quantity 
$$
\Delta=\frac{w^2}{\Vert \boldsymbol{u} \Vert_2^6} \left(
-3\gamma^2(\boldsymbol{u}^{\T} \boldsymbol{v})^2 
+ 2\gamma \Vert \boldsymbol{u} \Vert_2^2 (\boldsymbol{u}^{\T} \boldsymbol{v}) 
+ 4\gamma^2 \Vert \boldsymbol{u} \Vert_2^2 \Vert \boldsymbol{v} \Vert_2^2 
+ \Vert \boldsymbol{u} \Vert_2^4 \right)=\frac{w^2}{\Vert \boldsymbol{u} \Vert_2^6}\mathcal{Q}(\boldsymbol{u}^{\T}\boldsymbol{v}),
$$ 
is positive for all real vectors $\boldsymbol{u}, \boldsymbol{v}$, and it is strictly positive unless $\rho_{u,v} = -1$ and $\gamma \Vert \boldsymbol{v} \Vert_2 = \Vert \boldsymbol{u} \Vert_2$. 
\end{lemma}
\begin{proof}
Let $x = \boldsymbol{u}^{\T} \boldsymbol{v}\in\mathbb{R}$ and define $\mathcal{Q}(x) = -3\gamma^2 x^2 + 2\gamma \Vert \boldsymbol{u} \Vert_2^2 x + 4\gamma^2 \Vert \boldsymbol{u} \Vert_2^2 \Vert \boldsymbol{v} \Vert_2^2 + \Vert \boldsymbol{u} \Vert_2^4$. By the Cauchy-Schwarz inequality, $|x| = |\boldsymbol{u}^{\T} \boldsymbol{v}| \leq \Vert \boldsymbol{u} \Vert_2 \Vert \boldsymbol{v} \Vert_2$, then
\begin{align*}
\mathcal{Q}(x) &\geq -3\gamma^2 \Vert \boldsymbol{u} \Vert_2^2 \Vert \boldsymbol{v}_g \Vert_2^2 + 2\gamma \Vert \boldsymbol{u} \Vert_2^2 x + 4\gamma^2 \Vert \boldsymbol{u} \Vert_2^2 \Vert \boldsymbol{v} \Vert_2^2 + \Vert \boldsymbol{u} \Vert_2^4 \\
&= \gamma^2 \Vert \boldsymbol{u} \Vert_2^2 \Vert \boldsymbol{v} \Vert_2^2 + 2\gamma \Vert \boldsymbol{u} \Vert_2^2 x + \Vert \boldsymbol{u} \Vert_2^4.
\end{align*}
Since $-\Vert \boldsymbol{u} \Vert_2 \Vert \boldsymbol{v} \Vert_2\leq x \leq \Vert \boldsymbol{u} \Vert_2 \Vert \boldsymbol{v} \Vert_2$, even in the worst-case $x = -\Vert \boldsymbol{u} \Vert_2 \Vert \boldsymbol{v} \Vert_2$ or equivalently $\rho_{u,v} = -1$, we have:
\begin{align*}
\mathcal{Q}(\bu^{\T}\bv) &\geq \gamma^2 \Vert \boldsymbol{u} \Vert_2^2 \Vert \boldsymbol{v} \Vert_2^2 - 2\gamma \Vert \boldsymbol{u} \Vert_2^3 \Vert \boldsymbol{v} \Vert_2 + \Vert \boldsymbol{u} \Vert_2^4 \\
&= \Vert \boldsymbol{u} \Vert_2^2 \left( \gamma^2 \Vert \boldsymbol{v} \Vert_2^2 - 2\gamma \Vert \boldsymbol{u} \Vert_2 \Vert \boldsymbol{v} \Vert_2 + \Vert \boldsymbol{u} \Vert_2^2 \right) \\
&= \Vert \boldsymbol{u} \Vert_2^2 \left( \gamma \Vert \boldsymbol{v} \Vert_2 - \Vert \boldsymbol{u} \Vert_2 \right)^2 \geq 0,
\end{align*}
which completes the proof.
\end{proof}
\begin{lemma}
\label{lemma:traceAB_sufficient_condition}
Let $\bA \in \mathbb{R}^{n \times n}$ be a symmetric indefinite matrix, and let $\bB \in \mathbb{R}^{n \times n}$ be symmetric and positive semi-definite. Then $\trace(\bA\bB) > 0$, if
\begin{equation*}
\sum_{i\,:\,\lambda_i > 0} \lambda_i \bv_i^{\T} \bB \bv_i > \sum_{i\,:\,\lambda_i < 0} |\lambda_i| \bv_i^{\T} \bB \bv_i,
\end{equation*}
where $\{\bv_i\}$ is an orthonormal basis of eigenvectors and $\lambda_i \in \mathbb{R}$ are the corresponding eigenvalues of $\bA$.
\end{lemma}
\begin{proof}
Since $\bA$ is symmetric, it admits the spectral decomposition $\bA = \sum_{i=1}^n \lambda_i \bv_i \bv_i^{{\T}}$ with orthonormal eigenvectors $\{\bv_i\}$. Then,
\[
\bA\bB = \left( \sum_{i=1}^n \lambda_i \bv_i \bv_i^{{\T}} \right) \bB = \sum_{i=1}^n \lambda_i \bv_i (\bv_i^{\T} \bB),
\]
and taking the trace:
\[
\trace(\bA\bB) = \sum_{i=1}^n \lambda_i \trace(\bv_i \bv_i^{{\T}} \bB) = \sum_{i=1}^n \lambda_i \bv_i^{{\T}} \bB \bv_i.
\]
Now, separate the sum according to the sign of $\lambda_i$:
\[
\trace(\bA\bB) = \sum_{i\,:\,\lambda_i > 0} \lambda_i \bv_i^{{\T}} \bB \bv_i + \sum_{i\,:\,\lambda_i < 0} \lambda_i \bv_i^{{\T}} \bB \bv_i.
\]
Note that for $\lambda_i < 0$, we write $\lambda_i = -|\lambda_i|$, so:
\[
\trace(\bA\bB) = \sum_{i\,:\,\lambda_i > 0} \lambda_i \bv_i^{{\T}} \bB \bv_i - \sum_{i\,:\,\lambda_i < 0} |\lambda_i| \bv_i^{{\T}} \bB \bv_i,
\]
which is greater than zero, if and only if
\[
\sum_{i\,:\,\lambda_i > 0} \lambda_i \bv_i^{{\T}} \bB \bv_i > \sum_{i\,:\,\lambda_i < 0} |\lambda_i| \bv_i^{{\T}} \bB \bv_i,
\]
as required.
\end{proof}

\section{Empirical Assessment}
\label{sec:appendix_empirical_p>n}

We have investigated the $n<p$ regime separately. In this case a fundamental question arises: how should the weights of the adaptive procedures be defined when the least-squares estimator does not exist? To the best of our knowledge, the literature does not provide a generally accepted construction of the Adaptive Lasso or Adaptive Group Lasso in the $n<p$ regime, precisely because their definition relies on preliminary least-squares estimates.
As a pragmatic solution, we consider ridge regression estimates with minimal penalization as a surrogate preliminary estimator. We emphasize that this choice is not standard nor theoretically supported; it is simply a reasonable proposal in an otherwise ill-posed setting. We then evaluate empirically the performance of our proposed degrees of freedom estimators when plugging-in these quasi–least-squares weights.
Our simulations indicate that the analytical expressions derived under the $n>p$ framework, when combined with these ridge-based weights, yield degrees of freedom estimates that are close to the true values (computed numerically under a repeated-sampling principle). Although the adaptive procedure itself lacks a canonical definition in this regime, the empirical evidence suggests that our approach still substantially improves upon the common heuristic of approximating the degrees of freedom by the active set size.

\begin{figure}[h]
    \centering
    \includegraphics[width=0.75\linewidth]{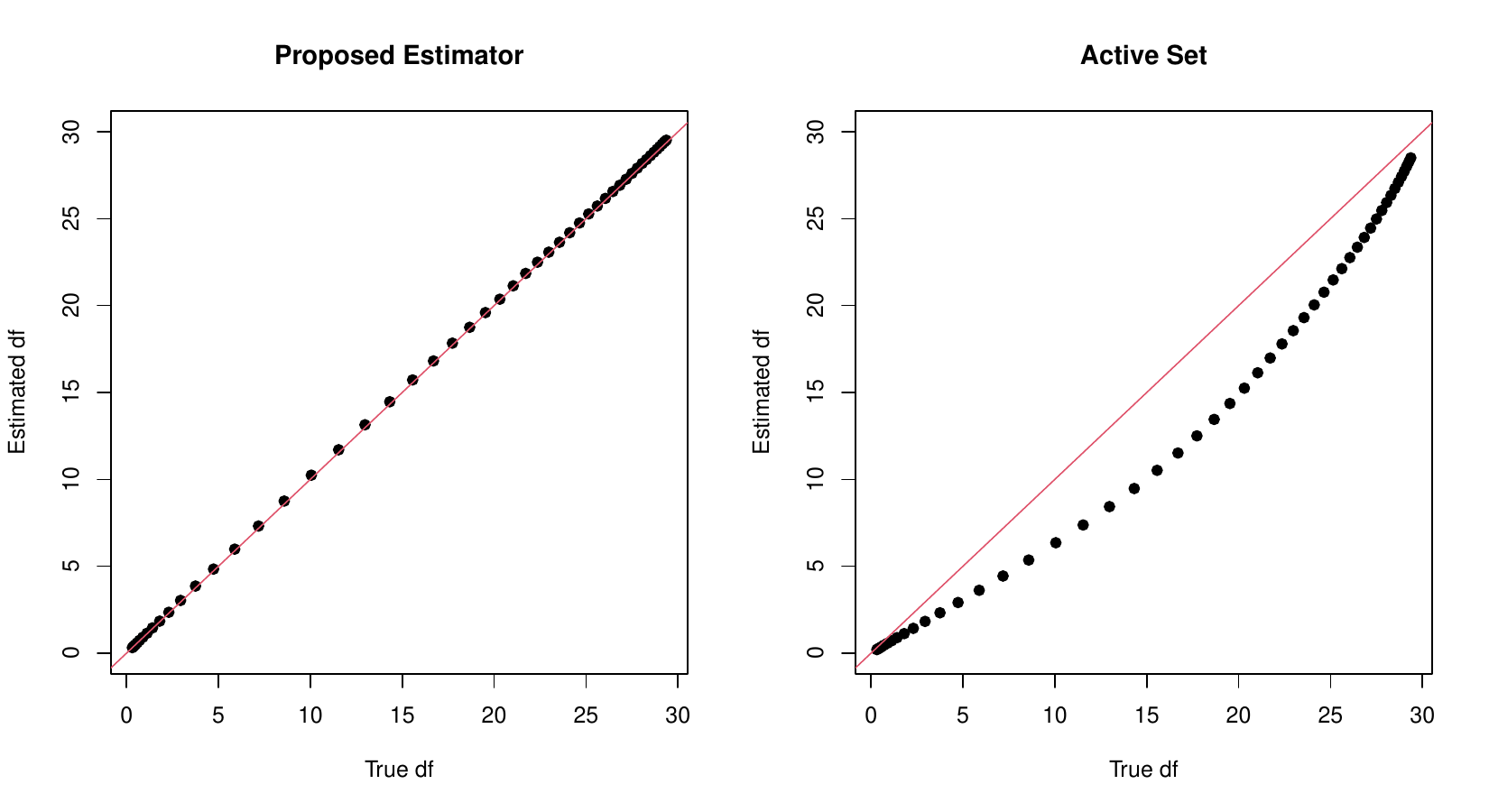}
    \caption{Estimated degrees of freedom versus degrees of freedom computed with the general covariance formula \eqref{eq:df_general_expression} for Adaptive Lasso. Left: estimator based on manuscript theorem. Right: estimator based on active set size.}
    \label{fig:p_greater_n}
\end{figure}


\end{document}